\PassOptionsToPackage{unicode}{hyperref}
\PassOptionsToPackage{hyphens}{url}
\PassOptionsToPackage{dvipsnames,svgnames,x11names}{xcolor}
\documentclass[
  12pt]{article}

\usepackage{amsmath,amssymb}
\usepackage{iftex}
\ifPDFTeX
  \usepackage[T1]{fontenc}
  \usepackage[utf8]{inputenc}
  \usepackage{textcomp} 
\else 
  \usepackage{unicode-math}
  \defaultfontfeatures{Scale=MatchLowercase}
  \defaultfontfeatures[\rmfamily]{Ligatures=TeX,Scale=1}
\fi
\usepackage{lmodern}
\ifPDFTeX\else  
\fi
\IfFileExists{upquote.sty}{\usepackage{upquote}}{}
\IfFileExists{microtype.sty}{
  \usepackage[]{microtype}
  \UseMicrotypeSet[protrusion]{basicmath} 
}{}
\makeatletter
\@ifundefined{KOMAClassName}{
  \IfFileExists{parskip.sty}{%
    \usepackage{parskip}
  }{
    \setlength{\parindent}{0pt}
    \setlength{\parskip}{6pt plus 2pt minus 1pt}}
}{
  \KOMAoptions{parskip=half}}
\makeatother
\usepackage{xcolor}
\makeatletter
\ifx\paragraph\undefined\else
  \let\oldparagraph\paragraph
  \renewcommand{\paragraph}{
    \@ifstar
      \xxxParagraphStar
      \xxxParagraphNoStar
  }
  \newcommand{\xxxParagraphStar}[1]{\oldparagraph*{#1}\mbox{}}
  \newcommand{\xxxParagraphNoStar}[1]{\oldparagraph{#1}\mbox{}}
\fi
\ifx\subparagraph\undefined\else
  \let\oldsubparagraph\subparagraph
  \renewcommand{\subparagraph}{
    \@ifstar
      \xxxSubParagraphStar
      \xxxSubParagraphNoStar
  }
  \newcommand{\xxxSubParagraphStar}[1]{\oldsubparagraph*{#1}\mbox{}}
  \newcommand{\xxxSubParagraphNoStar}[1]{\oldsubparagraph{#1}\mbox{}}
\fi
\makeatother

\usepackage{longtable,booktabs,array}
\usepackage{calc} 
\usepackage{etoolbox}
\makeatletter
\patchcmd\longtable{\par}{\if@noskipsec\mbox{}\fi\par}{}{}
\makeatother
\IfFileExists{footnotehyper.sty}{\usepackage{footnotehyper}}{\usepackage{footnote}}
\makesavenoteenv{longtable}
\usepackage{graphicx}
\makeatletter
\def\maxwidth{\ifdim\Gin@nat@width>\linewidth\linewidth\else\Gin@nat@width\fi}
\def\maxheight{\ifdim\Gin@nat@height>\textheight\textheight\else\Gin@nat@height\fi}
\makeatother
\setkeys{Gin}{width=\maxwidth,height=\maxheight,keepaspectratio}
\makeatletter
\def\fps@figure{htbp}
\makeatother

\makeatletter
\@ifpackageloaded{caption}{}{\usepackage{caption}}
\AtBeginDocument{%
\ifdefined\contentsname
  \renewcommand*\contentsname{Table of contents}
\else
  \newcommand\contentsname{Table of contents}
\fi
\ifdefined\listfigurename
  \renewcommand*\listfigurename{List of Figures}
\else
  \newcommand\listfigurename{List of Figures}
\fi
\ifdefined\listtablename
  \renewcommand*\listtablename{List of Tables}
\else
  \newcommand\listtablename{List of Tables}
\fi
\ifdefined\figurename
  \renewcommand*\figurename{Figure}
\else
  \newcommand\figurename{Figure}
\fi
\ifdefined\tablename
  \renewcommand*\tablename{Table}
\else
  \newcommand\tablename{Table}
\fi
}
\@ifpackageloaded{float}{}{\usepackage{float}}
\floatstyle{ruled}
\@ifundefined{c@chapter}{\newfloat{codelisting}{h}{lop}}{\newfloat{codelisting}{h}{lop}[chapter]}
\floatname{codelisting}{Listing}

\makeatother
\makeatletter
\@ifpackageloaded{caption}{}{\usepackage{caption}}
\@ifpackageloaded{subcaption}{}{\usepackage{subcaption}}
\makeatother

\ifLuaTeX
  \usepackage{selnolig}  
\fi
\usepackage[]{natbib}
\usepackage{bookmark}

\IfFileExists{xurl.sty}{\usepackage{xurl}}{} 
\hypersetup{
  pdftitle={Title},
  pdfauthor={Author 1; Author 2},
  pdfkeywords={3 to 6 keywords, that do not appear in the title},
  colorlinks=true,
  linkcolor={blue},
  filecolor={Maroon},
  citecolor={Blue},
  urlcolor={Blue},
  pdfcreator={LaTeX via pandoc}}

\usepackage{amsmath}
\usepackage{mathtools}

\usepackage{bm}
\usepackage{amsthm}
\usepackage{enumitem}
\usepackage{sidecap}
\usepackage{multirow}
\usepackage{float}

\usepackage{color}
\usepackage{statmath}
\usepackage{xfrac}
\usepackage{bigints}
\usepackage{bbm}
\usepackage{array}
\usepackage{booktabs}
\usepackage{siunitx, tabularx}
\usepackage{adjustbox}
\usepackage{arydshln,leftidx}
\usepackage{verbatim}
\usepackage{upgreek}
\usepackage{epstopdf}
\usepackage{bm}
\usepackage{makecell}  
\usepackage{bigints}
\usepackage{caption}

\usepackage[plain,noend]{algorithm2e}

\newtheorem{proposition}{Proposition}[section]
\newtheorem{theorem}{Theorem}[section]
\newtheorem{lemma}[theorem]{Lemma}

\newenvironment{namedproof}[1]{%
  \par\noindent\textit{#1.}\,%
}{\hfill$\square$\par}

\newcommand{\cR}
{\mathcal{R}}

\newcommand{\bW}
{\overline{W}}

\newcommand{\bP}
{\overline{P}}
\newcommand{\bR}
{\overline{R}}
\newcommand{\bcR}
{\overline{\mathcal{R}}}
\newcommand{\bS}
{\overline{S}}
\newcommand{\bs}
{\overline{s}}

\newcommand{\SR}
{\overline{S}_{\mathcal{R}}}
\newcommand{\sR}
{\overline{s}_{\mathcal{R}}}

\newcommand{\sRc}
{\overline{s}^{'}_{\mathcal{R}}}
\newcommand{\SRc}
{\overline{S}'_{\mathcal{R}}}
\newcommand{\sRl}
{s^{*}_{\cR,l}}
\newcommand{\SRl}
{S^{*}_{\cR,l}}

\newcommand{\PR}
{\overline{P}_{\mathcal{R}}}

\newcommand{\WR}
{\overline{W}_{\mathcal{R}}}

\newcommand{\Tree}{\mathrm{TREE}}

\DeclareRobustCommand{\rchi}{{\mathpalette\irchi\relax}}
\newcommand{\irchi}[2]{\raisebox{\depth}{$#1\chi$}}

\usepackage{amsmath,amsfonts,bm}

\newcommand{\nbracket}[1]{\left( #1 \right)}
\newcommand{\cbracket}[1]{\left\{ #1 \right\}}
\newcommand{\rbracket}[1]{\left[ #1 \right]}

\newcommand{\indep}{\perp\!\!\!\perp}

\newcommand{\anon}{1}

\begin{document}

\def\spacingset#1{\renewcommand{\baselinestretch}%
{#1}\small\normalsize} \spacingset{1}


\if1\anon
{
  \title{\bf Inference with Randomized Regression Trees}
  \author{
    Soham Bakshi\footnote{Equal contribution}\label{myUniqueLabel}
\newcounter{myfootnotecounter}
\setcounter{myfootnotecounter}{\value{footnote}}\\
    Department of Statistics,
    University of Michigan, MI, USA.\\
    and \\
    Yiling Huang\footnotemark[\value{myfootnotecounter}]\\
    Department of Statistics,
    University of Michigan, MI, USA.\\
    and \\
    Snigdha Panigrahi
    \\
    Department of Statistics,
    University of Michigan, MI, USA.\\
    and \\
    Walter Dempsey
    \\
    Department of Biostatistics and Institute for Social Research,\\
    University of Michigan, MI, USA.
  }
  \maketitle
} \fi

\if0\anon
{
  \bigskip
  \bigskip
  \bigskip
  \begin{center}
    {\LARGE\bf Title}
\end{center}
  \medskip
} \fi

\bigskip
\begin{abstract}
Regression trees are a popular machine learning algorithm that fit piecewise constant models by recursively partitioning the predictor space.  This paper focuses on statistical inference for a data-dependent model obtained from a fitted regression tree. We introduce Randomized Regression Trees (RRT), a novel selective inference method that adds independent Gaussian noise to the gain function underlying the splitting rules of classic regression trees.
The RRT method offers several advantages over existing methods. First, added randomization is used to obtain a closed-form pivot while accounting for the data-dependent tree structure. Second, RRT with a small amount of randomization achieves predictive accuracy similar to a model trained on the entire dataset, while also providing significantly more powerful inference than existing selective inference methods, such as data splitting. Third, RRT yields intervals that automatically adapt to the signal strength in the data. Our empirical analyses highlight these advantages of the RRT method and its ability to convert a purely predictive algorithm into a method capable of performing powerful inference in the non-linear tree model.
\end{abstract}

\noindent%
\vfill

\newpage
\spacingset{1.8} 

\section{Introduction}
\label{sec:intro}

Regression trees are a common machine learning algorithm for non-linear regression in which regions of the predictor space~$\mathcal{X}$ are recursively partitioned into smaller regions.
The outcome in these smaller regions is thought to be predicted well by a simple model~\citep{breiman1984classification}. The recursive partition is typically chosen to be a binary tree and splits of the binary tree correspond to half-spaces in the predictor space. 
Each terminal region in the tree represents a cell of the partition, and is accompanied by its own simple model. Regression trees have been widely applied across diverse scientific domains, including environmental health~\citep{Gass2014}, clinical and aging research~\citep{Allore2005}, ecology~\citep{Ndong2021}, and economics~\citep{ChernisSekkel2023}, for modeling complex non-linear relationships involving continuous responses.
Trees offer several advantages. First, predictions are fast and easy to compute. Second, piecewise constant models are good at approximating non-linear behavior. Third, these models are  fairly interpretable as the tree itself carries all the information necessary to tell what variables are important in forming predictions.
Despite these appealing features, classic regression trees largely remain an example of a \emph{pure prediction algorithm}~\citep{Efron2020_PEA}.
Such algorithms go directly for high predictive accuracy while neglecting both parameter estimation and attribution or 
significance of the estimated parameters in the models they fit.
Efron comments on this in his article:\\ 
``\emph{the pure prediction algorithms are a powerful addition to the statistician’s armory, yet substantial further development is needed for their routine scientific applicability}''.

In the context of regression trees, natural attribution questions include estimating the mean parameters in piecewise constant regression models, testing for differences in mean response between sibling terminal regions, and conducting goodness-of-fit tests to determine whether additional splits are necessary.
Since the sequence of greedily chosen splits leading to the tree fit are highly data dependent, na{\"i}ve approaches for inference, such as Wald-type tests and intervals, cannot be used for answering attribution questions in the fitted model.

Selective inference tools afford the ability to answer attribution questions in 
data dependent models.  
One of the simplest selective inference tools is data splitting: divide the data into two independent subsets, and then fit a regression tree to one subset and use the second subset for inference.
As an example, consider sample splitting, where the observations are partitioned into two disjoint subsets, assuming that they are independent and identically distributed.
A more recent variant of data splitting for a normal data vector is the {\it UV method} by \cite{rasines2023splitting}, which belongs to a larger category of data fission methods \citep{leiner2023data}.
A universal limitation of all data splitting approaches is the inherent trade-off between predictive modeling and inference, which can result in an inferior model compared to the na{\"i}ve method if too little data is used for modeling, or it may produce wide intervals if too little data is saved for inference. 
An alternative framework for conducting valid inference in data-dependent models, while using the full set of observations, including the data used during modeling, is conditional selective inference \citep{Lee_2016}.
\cite{neufeld2022tree} introduced such an approach for regression trees called {\it Tree-Values}, which uses the full dataset for both fitting the regression tree and conducting inference on data-dependent parameters.
While this method achieves the nominal coverage rate, like its precursor \cite{Lee_2016}, it produces intervals that can be much wider than those from data splitting approaches.

In this paper, we introduce a conditional selective inference methodology relying on a novel randomization scheme that adds independent Gaussian noise to the gain function underlying the splitting rule of classic regression trees.
This approach replaces greedy split selection with a softmax operation.
Based on this, we refer to our proposal as {\it Randomized Regression Trees (RRT)}. 
We show that by leveraging external randomization, our method addresses the limitations of both data-splitting approaches and the existing Tree-Values method.
The advantages of our approach are summarized below.
\begin{enumerate}[leftmargin=*]
    \item Our randomized procedure does not compromise predictive performance in order to deliver more powerful inference, or vice versa. In particular, we show that adding a small amount of randomization to the gain functions allows us to fit a tree model using almost the entire dataset, effectively mimicking the model fit with na{\"i}ve method. 
    
    Although similar predictive performance can be achieved through data splitting when most of the available data is allocated to training, the splitting approach leaves only a small portion of the data for inference. Whereas, our method is able to utilize the full dataset for inference, resulting in much shorter intervals and therefore more powerful inference than approaches using only the held-out portion not used at the time of model training.
    \item Our intervals, similar to those from conditional methods, automatically adapt to the amount of signal in the data, i.e.,
    our intervals widen or narrow depending on the strength of the effect of selection. Moreover, our empirical analyses demonstrate that even with a small amount of randomization, our method produces significantly shorter intervals compared to the Tree-Values method, which uses the full data for inference but does not incorporate external randomization. This highlights how the careful use of randomization can avoid the numerically unstable and excessively wide intervals that arise without it.
    \item Because of the specific randomization scheme we employ, the pivot underlying inference in the RRT method takes a convenient form: the greedy split selection via softmax is corrected for through a product of probabilities over independent randomization variables. Inference based on this pivot is numerically efficient, as accounting for the tree fit simplifies to evaluating $1$-dimensional integrals with respect to a univariate normal density.
\end{enumerate}

\begin{figure}[h]
    \centering
    \includegraphics{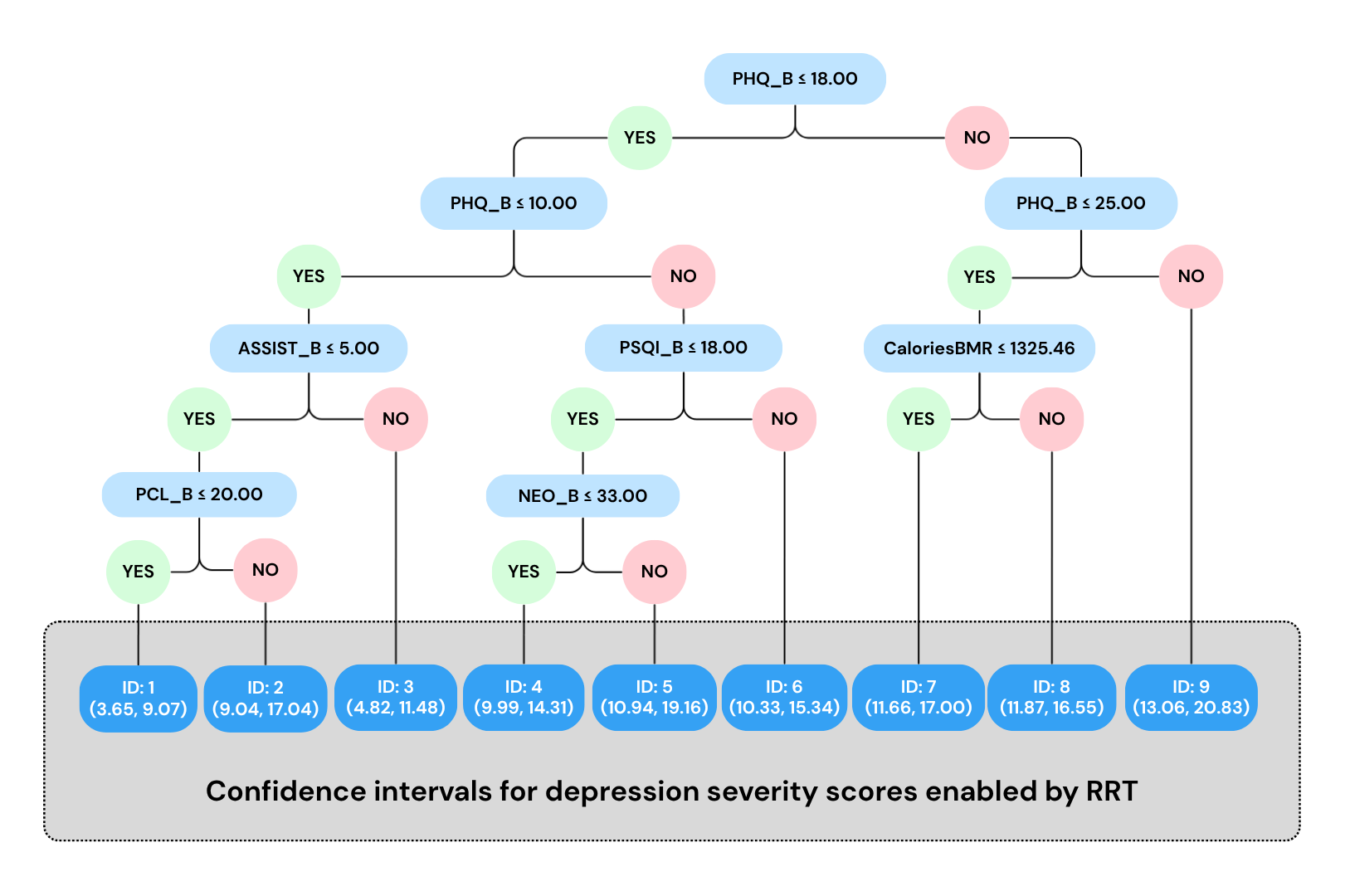}
    \captionsetup{width=\textwidth}
    \caption{Fitted RRT on the full PROMPT training set}
    \label{fig:PROMPT_TREE}
    \vspace{-0.5cm}
\end{figure}

Figure \ref{fig:PROMPT_TREE} presents the tree fit obtained with RRT on a real dataset predicting depression severity in a clinical mental health study. Standard tree-based methods, such as CART, can produce the tree fit but do not provide interval estimates for the population parameters associated with the tree structure.
With RRT, we provide confidence intervals for depression severity scores within the patient subpopulations learned with the tree fit, as illustrated by the intervals for each terminal region in the shaded box. Deferring a more detailed description of this real data analysis and interpretation of the results to Section \ref{sec:case study}, we next illustrate the practical advantages of our method over existing approaches through a worked example.

\section{Worked Example and Other Related Work}
\label{sec:firstexample}

\noindent{\textbf{Example}}. We consider a simple surface plus noise example using data generated in a similar fashion to \cite{neufeld2022tree}.  Let $X \in \mathbb{R}^{n \times p}$ with $n=200, p=5, X_{i j} \stackrel{i.i.d.}{\sim} \mathcal{N}(0,1)$, and $Y \sim \mathcal{N}_n\left(\mu, \sigma^2 I_n\right)$ with $\sigma=2$. The vector $\mu$ with $(\mu)_i=b \times\left[1_{\left(x_{i, 1} \leq 0\right)} \times\left\{1+a 1_{\left(x_{i, 2}>0\right)}+1_{\left(x_{i, 3} \times x_{i, 2}>0\right)}\right\}\right]$ defines a three-level tree where $a=1$, $b=2$ determine the signal strength.  

We consider our proposed RRT method and three baseline methods: (1) na{\"i}ve inference, (2) UV method (or Data Fission for normal data), (3) Tree-Values method.  For each method, the maximum depth of the final tree is 3, the minimum number of samples in a node to be split is 25, the minimum size of terminal nodes is 10, and we do not prune after the stopping criteria is met.  After fitting the trees, confidence intervals for terminal regions are computed and evaluated based on three metrics: (1) coverage rate, (2) average confidence interval length, and (3) test mean squared error (MSE) associated with the tree fit.    

In Figure~\ref{fig:toy_eg_1}, we first compare the two non-randomized inference methods, na{\"i}ve and Tree-Values, to our proposed method  on the same simulated dataset over 500 simulations.  
In our RRT proposal, the randomization sd $\tau$, specified in Section \ref{sec:simulation}, takes $4$ distinct values, corresponding to $4$ levels of randomization: RRT($1$), RRT($2$), $\ldots$, RRT($4$). Among these, RRT($1$) represents the method with the least amount of randomization. We make the following observations:
\begin{itemize}
    \item Both RRT and Tree-Values achieve valid $90\%$ coverage, while na{\"i}ve inference fails to deliver valid coverage. Tree-Values lead to very wide confidence intervals, whereas RRT, at all levels of randomization, produce significantly shorter intervals. On an average, the lengths of our intervals are roughly one-fifth of those produced by Tree-Values, corresponding to a substantial $80\%$ increase in inferential power over this earlier approach.
    \item Finally, in terms of predictive accuracy, RRT($1$) achieves comparable test MSE as the non-randomized methods. This shows that with a small amount of randomization, RRT converts a {\it pure prediction algorithm} into a method that achieves high predictive accuracy while also answering attribution questions in surface plus noise models with high inferential power.
\end{itemize}

\begin{figure}[h]
    \centering
    \includegraphics[width=\linewidth]{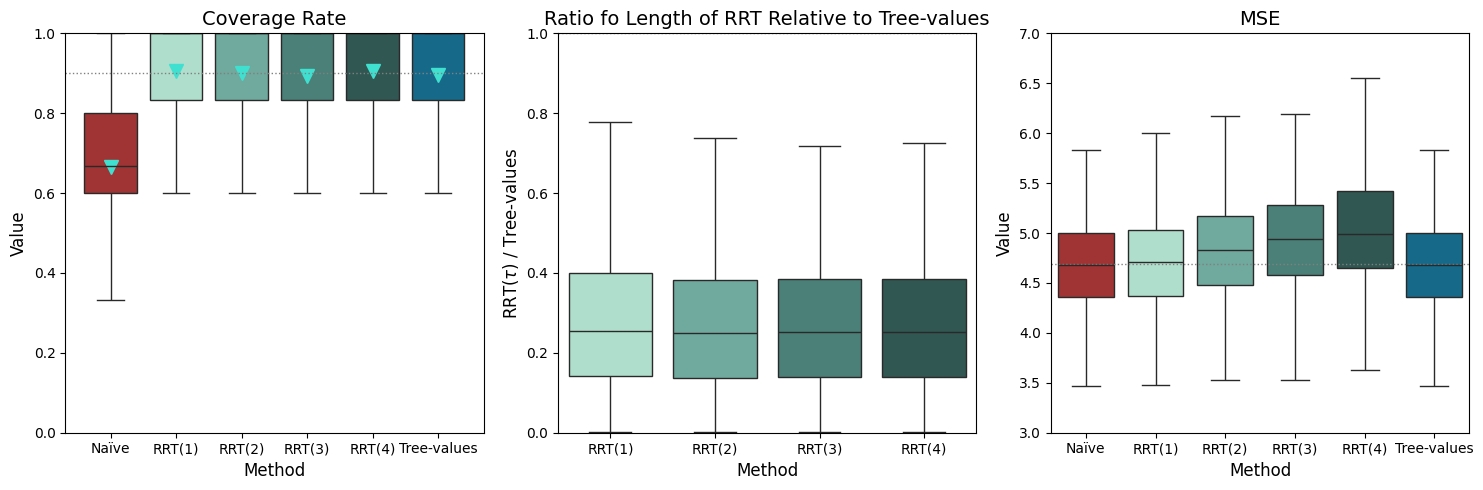}
    \vspace{-0.3cm}
    \captionsetup{width=\textwidth}
    \caption{Coverage rate, average CI length, and prediction MSE of the proposed method, Tree-Values, and Na{\"i}ve; {The dotted line is plotted at 0.9 in the coverage plot, and at the mean value of the proposed method in the MSE plot}; Triangles in the coverage plot indicates the empirical averages of coverage rates}
    \label{fig:toy_eg_1}
    \vspace{-0.5cm}
\end{figure}

We next compare RRT($1$), which achieves similar predictive accuracy as the na{\"i}ve method, to the UV method with $\gamma \in \{0.05, 0.1, 0.2, 0.3, 0.4, 0.5\}$, on the same simulated dataset. 
Figure~\ref{fig:toy_eg_2} presents the results for 500 simulations. Note that:
\begin{itemize}
    \item Both methods achieve $90\%$ nominal coverage. Though, there is a clear trade-off between attribution and predictive accuracy for the UV method.   Remarkably, even UV($0.5$), which sacrifices substantial predictive accuracy for inference, still produces intervals that are twice as long as those of RRT($1$).
    \item On the other hand, UV$(0.05)$, which has predictive accuracy similar to RRT($1$), produces confidence intervals that are nearly five times as wide as the RRT($1$) intervals.
\end{itemize}    

\begin{figure}[h]
    \centering
    \includegraphics[width=\linewidth]{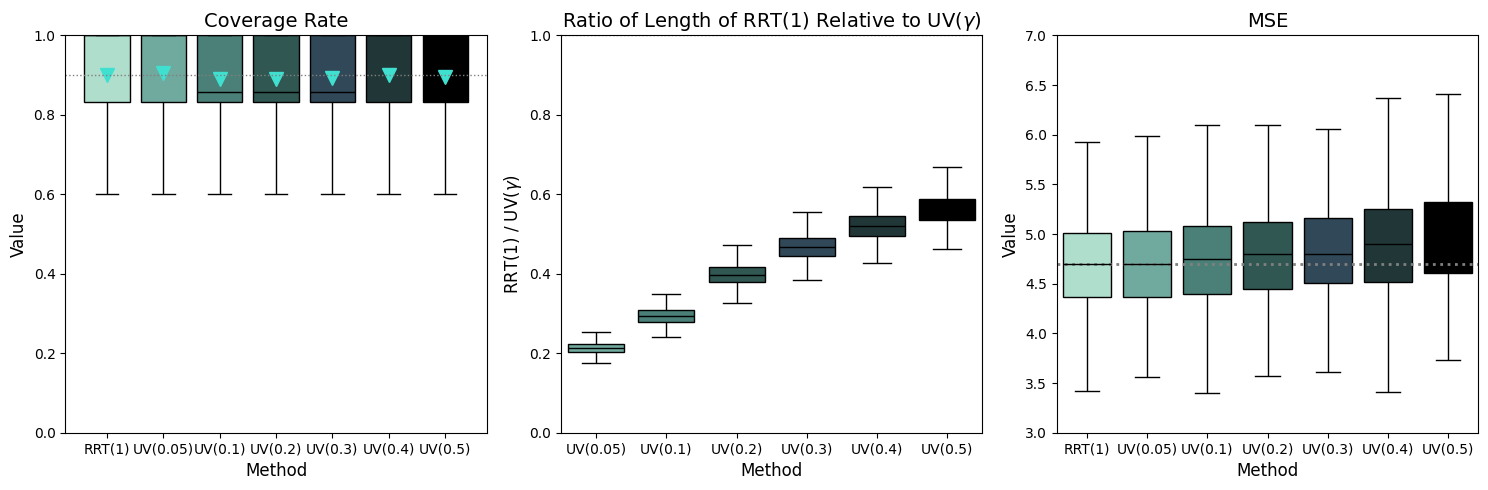}
    \vspace{-0.3cm}
    \captionsetup{width=\textwidth}
    \caption{Coverage rate, average CI length, and prediction MSE of the proposed method and the UV method; The dotted line is plotted at 0.9 in the coverage plot, and at the mean value of ``Na{\"i}ve'' in the MSE plot; Triangles in the coverage plot indicates the empirical averages of coverage rates}
    \label{fig:toy_eg_2}
    \vspace{-0.5cm}
\end{figure}

\noindent{\textbf{Other related work}}. The classification and regression trees (CART) is a decision-tree algorithm that builds a tree-like structure using splitting criteria \cite{breiman1984classification}. A model-based alternative called CTree was introduced by \cite{Hothorn01092006} in which the tree is grown by hypothesis testing.  
However, the CTree approach differs from our work in focus:  CTree focuses on fitting a tree model, whereas our approach addresses inferential questions that arise after a tree model has been fitted using our randomized tree-growing procedure.

In previous inferential work on regression trees, \cite{wager2016adaptiveconcentrationregressiontrees} developed convergence guarantees for the predictive surface of CART trees. 
Similarly, \cite{JMLR:v17:14-168} construct confidence intervals for random forest predictions based on subsampling, using the fact that these predictions can be expressed as U-statistics.
Note, however, that these approaches do not enable interval estimation for parameters associated with the tree fit---the inferential questions we refer to as attribution problems in the introduction.
Moreover, interestingly, our randomization approach employing a softmax operation for choosing splits opens the possibility for developing new random forest methods that do not rely on subsampling. 
\cite{Loh2018-hv} develop bootstrap calibration procedures that attempt to provide confidence intervals for the regions of a regression tree. However, this approach has been empirically shown to fail in providing intervals that achieve nominal coverage.

Attribution problems in regression trees naturally fit within the selective inference framework, as they involve data-dependent parameters rather than parameters specified a priori.
Our selective inference method takes a conditional approach to such problems.
Unlike simultaneous inference approaches that deliver guarantees over many plausible questions of interest \citep{Berk} or hybrid variants of the approach such as \cite{mccloskey2024hybrid}, the conditional approach yields inference for the specific model fitted to the data that practitioners ultimately rely on for interpreting their findings. Furthermore, by focusing on the fitted model observed in the data, the conditional approach---typically involving some randomization, as done in \cite{tian2018selective, Panigrahi02102023}---often yields substantially more powerful inference than simultaneous approaches to the same problem.
For a comparative analysis of these two approaches, we refer interested readers to \cite{perry2024infer}.

The choice of the conditioning event in a conditional approach, however, can lead to different types of inferential guarantees. 
As demonstrated in \cite{goeman2024selection}, a minimal conditioning event is needed to control the false coverage rate (FCR), which is the proportion of miscoverage over all constructed confidence intervals \citep{benjamini2005false}.
In this paper, our choice of conditioning event guarantees valid conditional inference for each mean parameter in the fitted tree model while also controlling the FCR across all fitted parameters, as shown in Section \ref{Sec:mainresults}.

Earlier work in conditional inference, such as \cite{charkhi2018asymptotic, jewell2022testing, guglielmini2025asymptotic, gao2024selective, perry2024inference}, employing no extra randomization, are known to have two main drawbacks. First, they can suffer from extremely low inferential power (i.e., infinitely wide intervals) when the observed data falls near the boundary of the selection event---a phenomenon we observe in our simulation examples with the non-randomized Tree-Values method. Second, explicit descriptions of conditioning events in terms of the data are not always available in closed form, sometimes even requiring computationally intensive Monte Carlo approaches for computing tests and interval estimates. Our randomized method bypasses both these difficulties by replacing the greedy split selection in CART with a simple softmax version through external randomization.
As already shown in Figure \ref{fig:toy_eg_1}, with a small amount of randomization, the RRT method does not trade off predictive power for inferential power, or vice versa, and produces substantially shorter intervals than the non-randomized Tree-Values method. 
Moreover, the additive form of randomization produces a closed-form pivot that is agnostic to the choice of information gain measure in the splitting criterion and enables numerically efficient inference by computing a series of fairly simple, one-dimensional integrals.

Although our work, like previous work in \cite{panigrahi2023approximate, huang2023selective, panigrahi2024exact, huang2023graphselective}, employs external randomization for inference, the RRT applies a different form of randomization than those methods, all which solve a linear regression fit. More specifically, these papers add a linear term involving both the unknown parameter and the randomization to the optimization objective. Whereas, by directly adding noise to the gain functions in the splitting criterion of the tree-growing algorithm, the splits in our tree fit are selected randomly according to probabilities determined by the randomization distribution. Finally, our setup assumes normally distributed data. Though we note that our pivot performs well even when the data deviate from normality. This robustness arises from the smoothness properties of our pivot constructed with randomization, as leveraged in \cite{panigrahi2023carving, bakshi2024selectiveinferencetimevaryingeffect} to perform selective inference without imposing parametric assumptions on the data distribution.
We therefore expect that our novel inferential approach can be naturally extended to provide valid inference for other data types in future work.

\section{RRT algorithm and model fit}
\label{sec: model}

We start with some basic notations. Consider a continuous response $Y\in \mathbb{R}^{n}$ and a set of $p$ predictors measured on $n$ observations, $X= (X^1, \ldots, X^p)$, where $X^{j} \in \mathbb{R}^{n}$ denotes the $j$-th predictor.
For a region $\cR \subseteq \mathbb{R}^{p}$, let $n_{\cR} = \lvert \cbracket{i\in [n]: X_{i} \in \cR}\rvert$ be the number of observations whose predictors fall within $\cR$ and let $\overline{Y}_{\cR}$ be the mean of these observations.

A standard TREE growing algorithm starts with the entire covariate space  $\mathbb{R}^{p}$, which we denote as $P_1$. The algorithm then recursively partitions this space with the goal of maximizing a certain notion of information gain based on the data $Y$.
Many popular tree-growing methods, including the CART, construct these partitions using a series of greedily selected binary splits. 
Splits take the form $s=(j,o)\in \rchi=[p] \times [n-1]$ where $j\in [p]$ denotes the index of the predictor selected for splitting and $o$ denotes the order statistic of this predictor.
For a non-terminal region $P \subseteq \mathbb{R}^{p}$, we let
$\rchi_{P} \subseteq \rchi$ denote the set of all possible splits that can be made on the region $P$.
Each split $s \in \rchi_{P}$ partitions $P$ into two half-spaces, $P_{s}^{l}=\left\{z \in P: (z)_j \leq X^{j}_{(o)}\right\}$, $P_{s}^{r}=\left\{z \in P: (z)_j>X^{j}_{(o)}\right\}$, where $(z)_j$ is the $j$-th coordinate of $z$ and $X^{j}_{(o)}$ is the $o$-th order statistic of the predictor $X^j$.
This split $s$ is associated with a measure of information gain, denoted by $G(Y; P, s)$.
In particular, the CART algorithm selects a split $s$ that maximizes the reduction in mean squared error (MSE) 
$$G(Y; P, s) =  \dfrac{1}{\sqrt{n_P}}\left[\sum_{i:X_{i} \in P} (Y_i - \overline{Y}_{P})^2 - \left\{\sum_{i:X_{i} \in P_{s}^{l}} (Y_i - \overline{Y}_{P_{s}^{l}})^2 + \sum_{i:X_{i} \in P_{s}^{r}} (Y_i - \overline{Y}_{P_{s}^{r}})^2 \right\}\right].$$ 
Equivalently, on each parent region $P$, it selects the optimal split as follows
$$
s^* = \underset{s\in \rchi_P}{\text{argmax}} \; G(Y; P, s)\equiv \underset{s\in \rchi_P}{\text{argmax}} \; \dfrac{1}{\sqrt{n_P}}\left\{n_{P^{l}_{s}}\overline{Y}_{P_{s}^{l}}^{2} + n_{P^{r}_{s}}\overline{Y}_{P^{r}_{s}}^{2}\right\},
$$
where the second equivalence follows by ignoring the constants in the first optimization objective.

\subsection{RRT algorithm}
\label{subsec:rrtbasic}

To grow a randomized regression tree (RRT), we modify the CART algorithm by adding independent normal random variables to the gain function at each split.
Specifically, for a parent region $P$, we maximize a randomized gain function to select the split $s^{*}$ from the set of possible splits $\chi_P$ as:
\begin{equation*}
  s^{*}  = \underset{s\in \chi_{P}}{\text{argmax}} \; G(Y; P, s)  + W(s), 
\end{equation*} 
where the external randomization variables $W(s) \sim \mathcal{N}(0, \tau_{P}^2)$ are independent of $Y$ and of each other, both at the same split and across different splits within the tree algorithm. This is equivalent to replacing the greedy split selection based on maximizing the information gain with a softmax.

Algorithm \ref{alg:RRT} outlines our method for constructing a RRT with a maximum depth $d_{\text{max}}$.
In the special case where $W(s)=0$, i.e.,  no randomization, the CART algorithm is recovered.
We start here with a fixed-depth tree, but in Section~\ref{app:adaptive} of the Supplementary Material we show how the RRT method readily extends to grow trees adaptively and perform inference in the fitted model.


\begin{algorithm}[h]
\caption{Randomized CART: fixed-depth tree growing algorithm}
\label{alg:RRT}
\KwIn{Training data $(X, Y)$, maximum depth $d_{\text{max}}$, noise variance $\tau_P^2$}
\KwOut{Set of terminal regions $\{\cR_1, \cR_2, \ldots, \cR_K\}$ forming the fitted randomized tree}

Initialize root region $P_1 \gets \mathbb{R}^p$, current depth $d \gets 0$\;

\While{$d < d_{\text{max}}$}{
    \ForEach{region $P$ at depth $d$ with $n_P > 1$}{
        Compute all candidate splits $\chi_P = \{(j, o)\}$\;
        Compute the impurity gain $G(Y; P, s)$ for all $s \in \chi_P$\;
        Draw random perturbations $W(s) \sim \mathcal{N}(0, \tau_P^2)$ for all $s \in \chi_P$\;
        Select randomized best split 
        $$s^* \gets \arg\max_{s \in \chi_P} \left\{ G(Y; P, s) + W(s) \right\}\;$$
        Partition region $P$ using $s^*$ into left and right children $\{P^{l}_{s^*}, P^{g}_{s^*}\}$\;
    }
    Increment tree depth: $d \gets d + 1$\;
}
Return all terminal regions $\{P\}$ as $\{\cR_1, \ldots, \cR_K\}$\;
\end{algorithm}

\subsection{TREE-model}
For developing inference, we introduce some more notation. Let $\bS$ be the collection of all splits generated by the RRT Algorithm \ref{alg:RRT}, and  $\bW$ is the collection of all external random variables added to the gain functions to generate these splits. 
The output of the RRT algorithm, denoted as $\Tree(Y, \bW)=\bR = \{R_1, R_2, \ldots, R_M\}$, is the set of terminal regions in the tree, also called leaves.
Furthermore, let $\bP = \cbracket{P_{1}=\mathbb{R}^{p},\ldots,P_{K}}$ denote all the internal or non-terminal regions in this tree output. 

Suppose for our observed realization of the data $Y=y$, the $M$ terminal regions observed as output of the RRT algorithm are $\{R_{1}=\cR_1, R_{2}=\cR_2, \ldots, R_{M}=\cR_M\}.$
Associated with this output, a simple predictive piecewise-constant TREE-model for the RRT is given by: 
\begin{equation}
Y\sim \mathcal{N}(\mu, \sigma^{2} I_n), \text{ where } (\mu)_{i} = \sum_{m=1}^{M} \mu_{\cR_{m}} \mathbbm{1}[X_{i} \in \cR_{m}],  \ \forall i\in [n],
\label{tree:model}
\end{equation}
where $X_i$ represents the $i$-th observed value for the $p$ predictors. 
Note that the $M$ parameters $\{\mu_{\cR_1}, \mu_{\cR_2}, \ldots, \mu_{\cR_m}\}$ in this TREE-model depend on the data and randomization variables through the splits generated during the tree-growing process of the RRT algorithm. 

This model is commonly used for predicting outcomes for new observations, yet an equally fundamental task is to draw statistical inference on the population parameters of the fitted model.
One natural task is performing inference on the mean parameters $\mu_{\cR}$, for $\cR \in \bcR$, where $\bcR$ is the observed value of $\bR$. 
Figure~\ref{fig:toy_eg_1} in Section \ref{sec:intro} shows the interval estimates of the mean parameters in the model fitted by the RRT.
Another important task is comparing differences between the mean parameters of two sibling regions within the tree, $\mu_{\cR}-\mu_{\cR'}$, if $\cR$ and $\cR'$ are two sibling regions in the TREE-model. 
As emphasized earlier, none of these tasks can be addressed with na{\"i}ve inference, due to the complex dependence between the TREE-model and the observed data.


\section{Exact pivot for inference}
\label{Sec:mainresults}

In this section, we focus on fixed-depth trees and develop inferential results for the mean parameter associated with an observed terminal region $\cR \in \bcR$. We provide a simplified overview of our inferential framework using a one-depth (single-split) tree in Section~\ref{app:onesplit} of the Supplementary Materials. 
Extensions of our method and their theoretical guarantees for a TREE-model fitted with adaptive stopping rules are provided in Section~\ref{app:adaptive} of the Supplementary Materials. 

To track variables in the subtree leading to the observed terminal region $\cR$, let $\PR = \left\{P_{\cR,1}=\mathbb{R}^p, P_{\cR,2}, \ldots, P_{\cR,L} \right\}$ 
denote the set of parent regions that were recursively split to obtain $\cR$, i.e., the last split on $P_{\cR,L}$ results in the terminal region $\cR$.
Let $\chi_{\cR, l}$ denote the set of possible splits at the parent region $P_{\cR,l}$, where $|\chi_{\cR, l}|=d_l$, and 
let
$$\SR =  \cbracket{S^{*}_{\cR,1}, S^{*}_{\cR,2}, \ldots,S^{*}_{\cR,L}}$$
be the series of $L$ random splits made on this sequence of parent regions in $\PR$ leading to $\cR$.
As before, we let $\sR= \cbracket{s^{*}_{\cR, 1}, s^{*}_{\cR, 2}, \ldots, s^{*}_{\cR, L}}$ denote the observed values of these splits, i.e., for the specific TREE-model realized when $Y=y$, we observe
$\{\SR =\sR\}$.
Furthermore, we let $\SRc = \bS \setminus \SR$ be the set of splits in the RRT that are not included in the subtree leading to $\cR$, and we let $\sRc$ be the observed values of these splits.
Recall that at each split in this subtree, the RRT Algorithm \ref{alg:RRT} adds independent randomization variables to the gain function. 
We collect these variables and denote them by
$\WR = \cbracket{\bW_{\cR,1},\ldots, \bW_{\cR,L}}$ where 
$\bW_{\cR, l} = \cbracket{W_{\cR, l}(s), \; \forall s \in \rchi_{\cR, l}}$.
For the reader's convenience, these notations are summarized in Table \ref{Table:notations2}.

\renewcommand{\arraystretch}{1.2}
{
\begin{table}[h!]
\centering
\begin{tabular}{ccc}
\toprule
 \makecell{Variable/Object} & \makecell{Definition} &  \makecell{Set/Vector}  \\
\midrule
$\PR$ & Parent regions of $\cR$ & $\cbracket{P_{\cR,1}, P_{\cR,2}, \ldots,P_{\cR,L}}$ \\
$\SR$ &  Splits leading to $\cR$ & $\cbracket{S^{*}_{\cR,1}, S^{*}_{\cR,2}, \ldots,S^{*}_{\cR,L}}$ \\
$\WR$ & \makecell{Randomization variables \\ linked with $\SR$} & $\cbracket{W_{\cR,1}, W_{\cR,2}, \ldots,W_{\cR,L}}$  \\
$\SRc$ &  Splits not relevant to $\cR$ & $\bS \setminus \SR$ \\
$\overline{D}_{\cR}$ & \makecell{Differences in gain functions\\ linked with $\cR$} &$\nbracket{D_{\cR,l}: l \in [L] }$  \\
\bottomrule
\end{tabular}
\caption{A list of notations for variables/objects in the subtree leading to an observed terminal region $\cR$}
\label{Table:notations2}
\end{table}
}

\subsection{Conditional distribution of interest and guarantees}

Our main result in this section, Theorem \ref{pivot:1}, provides the expression for an exact pivot, which is a function of the observed data $Y=y$ and our parameter of interest, $\mu_{\cR}$. 
Let $\nu_{\cR}$ be an $n$-dimensional vector such that 
$$(\nu_{\cR})_{i} = \dfrac{\mathbbm{1}[X_{i}\in \cR]}{\sqrt{n_{\cR}}}.$$
It is straightforward to see that, in the fitted TREE-model in \eqref{tree:model}, $\nu_{\cR}^\top\mu = \sqrt{n_{\cR}} \mu_{\cR}$.
In what follows, we develop a pivot for $\nu_{\cR}^\top\mu$.

Following the principles of conditional inference, we base inference on the conditional distribution of
$\nu_{\cR}^\top Y  \mid P_{\nu_{\cR}}^{\perp}Y = P_{\nu_{\cR}}^{\perp}y, \bS = \bs,$
where $\bS$ is the collection of all splits in $\Tree(Y, \bW)$, $\nu_{\cR}^\top Y$ is the sample mean of observations landing in $\cR$, scaled by $\sqrt{n_{\cR}}$, and 
$$ 
P_{\nu}^{\perp}Y= \left(I_n -  \frac{\nu \nu^{\top}}{||\nu||_{2}^{2}} \right)Y
$$
is the projection on to the orthogonal subspace spanned by the vector $\nu_{\cR}$.
Note that the na{\" i}ve Wald intervals are based on the Gaussian distribution of $\nu_{\mathcal{R}}^\top Y$, ignoring the dependence of $\cR$ on the observed data.

The exact pivot facilitates the construction of confidence intervals $\left\{\widehat{C}_{\cR} : \cR \in \bcR\right\}$, where $\widehat{C}_{\cR}$ is the interval estimate for $\mu_{\cR}$.
Noting that $\{\bS = \bs\} \subseteq \{R_{1}=\cR_1,  \ldots, R_{M}=\cR_M\}= \{\bR =\bcR\}$, the simple rationale behind conditioning on the splits is:
\begin{equation}
\label{cond:guarantee}
\begin{aligned}
\mathbb{P}\left[\mu_{\cR} \in \widehat{C}_{\cR}\mid \bS = \bs \right] \geq 1-\alpha 
&\implies \mathbb{P}\left[\mu_{\cR} \in \widehat{C}_{\cR}\mid \bR=\bcR \right] \geq 1-\alpha\\
&\implies \mathbb{P}\left[\mu_{\cR} \in \widehat{C}_{\cR} \right]\geq 1-\alpha.
\end{aligned}
\end{equation}
The implications in \eqref{cond:guarantee} are a direct consequence of the tower property of expectation.
The conditional guarantee on the left-hand side ensures coverage at the desired level for each individual parameter $\mu_{\cR}$ for $\cR\in \bcR$.
Moreover, following the same argument as \cite{Lee_2016}, it also controls the false coverage rate (FCR), defined in \cite{benjamini2005false}, at level $\alpha$, i.e., 
$$
\text{FCR} = \mathbb{E}\left[\dfrac{\Big|\left\{j\in [M]: \mu_{\cR_{j}} \notin \widehat{C}_{\cR_{j}} \right\}\Big|}{\max(M,1)} \right] \leq \alpha.
$$

\subsection{Exact pivot in closed-form}

To state our main result introducing the pivot, we first establish a few useful results analyzing  the conditional probability of the selection event given the observed data. 



Lemma \ref{L1} states that, to obtain the conditional density of interest, we only need to account for the selection of the splits in $\SR$, which form the subtree leading to $\cR$.

\begin{lemma}
\label{L1}
For an observed terminal region $\cR \in \bcR$, it holds that 
$$\mathbb{P}\rbracket{\bS = \bs \mid P_{\nu_{\cR}}^{\perp}Y = P_{\nu_{\cR}}^{\perp}y, \nu_{\cR}^\top Y  = t } \propto \mathbb{P}\rbracket{\SR = \sR \mid \nu_{\cR}^\top Y  = t, P_{\nu_{\cR}}^{\perp}Y =   P_{\nu_{\cR}}^{\perp}y},$$
when the probability on the left-hand side of the display is evaluated as a function of $\nu_{\cR}^\top Y = t \in \mathbb{R}$.
\end{lemma}


Proposition \ref{prop:1} uses Lemma \ref{L1} to obtain the conditional density of interest.

\begin{proposition}
\label{prop:1}
The conditional density of $\nu_{\cR}^\top Y \mid \cbracket{\bS = \bs, P_{\nu_{\cR}}^{\perp}Y = P_{\nu_{\cR}}^{\perp}y}$, when evaluated at $\nu_{\cR}^\top Y= t$, equals 
$$\dfrac{\phi \nbracket{t; \nu_{\cR}^\top\mu, \sigma^{2}||\nu_{\cR}||^2_2 } \times  \mathbb{P}\rbracket{\SR = \sR \mid \nu_{\cR}^\top Y  = t, P_{\nu_{\cR}}^{\perp}Y =   P_{\nu_{\cR}}^{\perp}y}}{\int\phi \nbracket{t'; \nu_{\cR}^\top\mu, \sigma^{2}||\nu_{\cR}||^2_2 } \times  \mathbb{P}\rbracket{\SR = \sR \mid \nu_{\cR}^\top Y  = t', P_{\nu_{\cR}}^{\perp}Y =   P_{\nu_{\cR}}^{\perp}y }dt'}.$$
\end{proposition}



Our next result, Proposition \ref{prop:cond:density}, provides an expression for the probability of the selected splits $\{\SR = \sR\}$, given data, as an explicit function of  $(\nu_{\cR}^{\top}Y, P_{\nu_{\cR}}^{\perp}Y)$. This forms the basis of our pivot.

For the series of $L$ selected splits $\sR=\cbracket{s^{*}_{\cR,1}, s^{*}_{\cR,2}, \ldots, s^{*}_{\cR,L}}$, define for each $l \in [L]$:
\begin{align}
\begin{gathered}
\beta_{\sRl}(t, P_{\nu_{\cR}}^{\perp}y) = \left( G\nbracket{y(t) ;P_{\cR,l}, s} -G\nbracket{y(t) ;P_{\cR,l},\sRl}:  s\in \rchi_{\cR,l}\setminus \{\sRl\}\right)\in \mathbb{R}^{d_l-1},\\
\Omega_{\cR,l} = \tau_{P_{\cR,l}}^2 
\nbracket{I_{d_l-1} + 1_{d_l-1}1_{d_l-1}^{\top}}\in \mathbb{R}^{(d_l-1) \times (d_l-1)},
\end{gathered}
\label{eq:mean_opt_var}
\end{align}
where 
$$y(t) =  t\dfrac{\nu}{||\nu||_{2}^{2}} + P_{\nu_{\cR}}^{\perp}y$$
is the response reconstructed with $P_{\nu_{\cR}}^{\perp}Y = P_{\nu_{\cR}}^{\perp}y$, the observed value from our data, and $\nu_{\cR}^\top Y = t$.
The vector $\beta_{\sRl}(t, P_{\nu_{\cR}}^{\perp}y)$ collects the difference in the information gains between the optimal split $\sRl$ and the losing splits $s\in \rchi_{\cR,l}\setminus \{\sRl\}$ that were not chosen.
Then, we define
\begin{align*}
\begin{gathered}
\Lambda_{\sRl}(t, P_{\nu_{\cR}}^{\perp}y) = \int_{\mathbb{R}_{+}^{d_{l}-1}} \phi\nbracket{u; \beta_{\sRl}(t, P_{\nu_{\cR}}^{\perp}y), \Omega_{\cR,l}} du,
\end{gathered}
\end{align*}
a function of $(t, P_{\nu_{\cR}}^{\perp}y)$. 
Note that $\Lambda_{\sRl}(t, P_{\nu_{\cR}}^{\perp}y)$ is a Gaussian integral over the $(d_l-1)$-dimensional positive orthant.

\begin{proposition}
\label{prop:cond:density}
Given an observed terminal region $\cR \in \bcR$, we have that 
$$\mathbb{P}\rbracket{\SR = \sR \mid \nu_{\cR}^{\top}Y =t, P_{\nu_{\cR}}^{\perp}Y=P_{\nu_{\cR}}^{\perp}y} = \prod_{l=1}^{L} \Lambda_{\sRl}(t, P_{\nu_{\cR}}^{\perp}y).$$
\end{proposition}

 \begin{proof}
We begin by noting that
\begin{equation}
 \begin{aligned}
     \cbracket{\SR= \sR} 
     =&\cbracket{S^{*}_{\cR,1}= s^{*}_{\cR,1}, S^{*}_{\cR,2}=s^{*}_{\cR,2}, \ldots, S^{*}_{\cR,L}=s^{*}_{\cR,L}} \\
     =& \bigcap_{l=1}^{L}\cbracket{\SRl= \sRl} \\
     =& \bigcap_{l=1}^{L}\Big\{G(Y;P_{\cR,l},  \sRl) - G(Y;P_{\cR,l},s) + W_{\cR, l}(\sRl) - W_{\cR, l}(s) \geq 0, \ \\
     &\;\;\;\;\;\;\;\;\;\;\;\;\;\;\;\;\;\;\;\;\;\;\;\;\;\;\;\;\;\;\;\;\;\;\;\;\;\;\;\;\;\;\;\;\;\;\;\;\;\;\;\;\;\;\;\;\;\;\;\;\;\;\;\;\; \;\;\;\;\;\;\;\;\;   \forall s \in \rchi^{*}_{\cR,l}\setminus \{\sRl\}\Big\} \\
     =& \bigcap_{l=1}^{L}\Big\{W_{\cR, l}(\sRl) - W_{\cR, l}(s) \geq  G(Y;P_{\cR,l},s) - G(Y;P_{\cR,l},  \sRl), \\
     & \;\;\;\;\;\;\;\;\;\;\;\;\;\;\;\;\;\;\;\;\;\;\;\;\;\;\;\;\;\;\;\;\;\;\;\;\;\;\;\;\;\;\;\;\;\;\;\;\;\;\;\;\;\;\;\;\;\;\;\;\;\;\;\;\; \;\;\;\;\;\;\;\;\;   \forall s \in \rchi^{*}_{\cR,l}\setminus \{\sRl\}\Big\}.
 \end{aligned}
 \label{event:description}
\end{equation} 
Computing the probability of this event, it holds that 
 \begin{align*}
& \mathbb{P}\rbracket{\SR = \sR \mid \nu_{\cR}^{\top}Y =t, P_{\nu_{\cR}}^{\perp}Y=P_{\nu_{\cR}}^{\perp}y}\\
     =& \prod_{l=1}^{L} \mathbb{P}\rbracket{\SRl= \sRl \mid \nu_{\cR}^{\top}Y =t, P_{\nu_{\cR}}^{\perp}Y=P_{\nu_{\cR}}^{\perp}y} \\
     =& \prod_{l=1}^{L} \mathbb{P}\Big[W_{\cR, l}(\sRl) - W_{\cR, l}(s) \geq  G(Y;P_{\cR,l},s) - G(Y;P_{\cR,l},  \sRl), \\
     & \;\;\;\;\;\;\;\;\;\;\;\;\;\;\;\;\;\;\;\;\;\;\;\;\;\;\;\;\;\;\;\;\;\;\;\;\;\;\;\;\;\;\;\;\;\;\;\;\;\;\;\;\;\;\;\;\;\;\;\;\; \forall s \in \rchi^{*}_{\cR,l} \setminus \{\sRl\} \mid \nu_{\cR}^{\top}Y =t, P_{\nu_{\cR}}^{\perp}Y=P_{\nu_{\cR}}^{\perp}y \Big] \\
      =& \prod_{l=1}^{L} \mathbb{P}\Big[W_{\cR, l}(\sRl) - W_{\cR, l}(s) \geq  G(y(t);P_{\cR,l},s) - G(y(t);P_{\cR,l},  \sRl), \ \forall s \in \rchi^{*}_{\cR,l} \setminus \{\sRl\}\Big] \\ 
     =&  \prod_{l=1}^{L} \Lambda_{\sRl}(t, P_{\nu_{\cR}}^{\perp}y). \tag{5}\label{integral:form}
 \end{align*}

Here, the first display is due to the independence between the randomization variables at different splits, and as a result, we have that $S^{*}_{j} \indep S^{*}_{j'} \mid Y(t) = y(t)$ for $j\neq j'$. The second display uses the description of the event in \eqref{event:description}. The third display uses the independence between the external randomization variables and $Y$. We arrive at the final display by noting that $$\nbracket{W_{\cR, l}(\sRl) - W_{\cR, l}(s): s \in \rchi^{*}_{\cR,l}\setminus \{\sRl\}} \in \mathbb{R}^{d_{l}-1}$$ is distributed as a normal random variable with mean $0_{d_l-1}$ and covariance $\Omega_{\cR,l}$.

In the last step, letting
$$T_{\sRl}=\left\{ t\in \mathbb{R}^{d_l-1}: (t)_j \geq G(y(t);P_{\cR,l},s) - G(y(t);P_{\cR,l},  \sRl)\ \forall j \in [d_l-1]\right\},$$
we observe that
\begin{equation*}
\begin{aligned}
& \mathbb{P}\Big[W_{\cR, l}(\sRl) - W_{\cR, l}(s) \geq  G(y(t);P_{\cR,l},s) - G(y(t);P_{\cR,l},  \sRl), \ \forall s \in \rchi^{*}_{\cR,l} \setminus \{\sRl\}\Big] \\
& =\int_{T_{\sRl}} \phi(u'; 0_{d_l-1}, \Omega_{\cR, l})du' \\
&= \int_{\mathbb{R}_{+}^{d_{l}-1}} \phi\nbracket{u; \beta_{\sRl}(t, P_{\nu_{\cR}}^{\perp}y), \Omega_{\cR,l}} du\\
&=  \Lambda_{\sRl}(t, P_{\nu_{\cR}}^{\perp}y).
\end{aligned}
\end{equation*}
This completes the proof.
 \end{proof}
 
At last, let
$$\Lambda_{\sR}(t, P_{\nu_{\cR}}^{\perp}y)=\prod_{l=1}^{L} \Lambda_{\sRl}(t, P_{\nu_{\cR}}^{\perp}y).$$ 
In Theorem \ref{pivot:1}, we present the final expression for an exact pivot, derived using the conditional density from Proposition \ref{prop:1} and plugging into it the probability expression from Proposition \ref{prop:cond:density}.

\begin{theorem}
Given data $Y=y$, let 
$$P(\nu_{\cR}^\top y,P_{\nu_{\cR}}^{\perp}y; \nu_{\cR}^\top \mu) = \frac{\int_{-\infty}^{\nu_{R}^\top y} \phi(t; \nu_{\cR}^\top \mu, \sigma^2||\nu_{\cR}||^2_2) \Lambda_{\sR} (t, P_{\nu_{\cR}}^{\perp}y) dt}{\int_{-\infty}^{+\infty}\phi(t; \nu_{\cR}^\top \mu, \sigma^2||\nu_{\cR}||^2_2) \Lambda_{\sR} (t, P_{\nu_{\cR}}^{\perp}y)dt}.$$
Then, it holds that
$$P(\nu_{\cR}^\top Y,P_{\nu_{\cR}}^{\perp}Y; \nu_{\cR}^\top \mu)\; \lvert \; \bS=\bs \sim \text{\normalfont Unif}(0,1).$$
\label{pivot:1}
\end{theorem}

The above result implies that for our observed data $Y=y$, the pivot for $\nu_{\cR}^\top \mu$ equals $P(\nu_{\cR}^\top y,P_{\nu_{\cR}}^{\perp}y; \nu_{\cR}^\top \mu)$.
This pivot ensures valid inference in the RRT, achieving the nominal coverage rate while also controlling the FCR over the  terminal regions fit with the RRT.

The proposed pivot from Theorem \ref{pivot:1} calculated under the null $\nu_{\cR}^\top \mu = 0$ serves as valid p-value for testing. Moreover, to obtain confidence intervals for the selected target $\nu_{\cR}^\top \mu$, we simply invert our pivot. For example, two-sided confidence intervals at level $\alpha$ are constructed as
$$
\left(L_{\cR}, U_{\cR}\right)=\left\{\nu_{\cR}^\top \mu: P(\nu_{\cR}^\top y,P_{\nu_{\cR}}^{\perp}y; \nu_{\cR}^\top \mu) \in\left[\frac{\alpha}{2}, 1-\frac{\alpha}{2}\right]\right\}.
$$
We also note that the same idea can be directly extended to construct a pivot for the difference between the mean parameters of sibling regions,  $\cR$ and $\cR'$.

\subsection{Numerically efficient inference using the exact pivot}

In Theorem \ref{pivot:1}, $\Lambda_{\sR}(t, P_{\nu_{\cR}}^{\perp}y)$, for $t \in \mathbb{R}$, is the correction function ensuring valid inference in the RRT.
Multiplying this correction function with the na{\"i}ve density of $\nu_{\cR}^\top Y$ takes into account the data-dependent nature of $\cR$, the terminal region linked to our parameter of interest $\mu_{\cR}$. 

Revisiting the correction function, we see that it decomposes into a product of $L$ integrals over the independent randomization variables at $L$ splits, as $\Lambda_{\mathcal{R}}(t, P_{\nu_{\mathcal{R}}}^{\perp}y)= \prod_{l=1}^{L} \Lambda_{\sRl}(t, P_{\nu_{\cR}}^{\perp}y)$.  
Each integral in this product, $\Lambda_{\sRl}(t, P_{\nu_{\cR}}^{\perp}y)$, is $(d_l - 1)$-dimensional, where $d_l = |\chi_l|$.
Computing these integrals in this form would be computationally burdensome, particularly when the number of candidate splits $d_l$ on parent region $P_{\cR,l}$ is large.

However, as our next result in Proposition \ref{prop:1d-intergral} shows, due to our specific form of randomization, the correction at each split can be rewritten as a $1$-dimensional integral. 

\begin{proposition}
    \label{prop:1d-intergral}
    It holds that 
    $$
    \Lambda_{\sRl}(t, P_{\nu_{\cR}}^{\perp}y)=\int_{\mathbb{R}} \Gamma_{\sRl}(w; t, P_{\nu_{\cR}}^{\perp}y)\phi(w;0,1) dw,
    $$
    where 
    $\Gamma_{\sRl}(w; t, P_{\nu_{\cR}}^{\perp}y)= \displaystyle\prod_{s\in \rchi^{*}_{\cR,l}\setminus \{\sRl\}} \Phi\left(\dfrac{w-\left[\beta_{\sRl}(t, P_{\nu_{\cR}}^{\perp}y)\right]_s}{\tau_{P_{\cR,l}}}\right)$ , and 
    $$\left[\beta_{\sRl}(t, P_{\nu_{\cR}}^{\perp}y)\right]_s = G(y(t);P_{\cR,l},s) - G(y(t);P_{\cR,l},  \sRl)$$ is the $s$-th entry of $\beta_{\sRl}(t, P_{\nu_{\cR}}^{\perp}y)$ corresponding to split $s\in \rchi^{*}_{\cR,l}\setminus \{\sRl\}$, and $\Phi(\cdot)$ is the CDF of the standard normal distribution.
\end{proposition}
\begin{proof}
Notice from \eqref{integral:form} that $\Lambda_{\sRl}(t, P_{\nu_{\cR}}^{\perp}y)$ is equal to 
\begin{align*}
 & \mathbb{P}\Big[W_{\cR, l}(\sRl) - W_{\cR, l}(s) \geq  G(y(t);P_{\cR,l},s) - G(y(t);P_{\cR,l},  \sRl), \ \forall s \in \rchi^{*}_{\cR,l} \setminus \{\sRl\}\Big].
\end{align*}    
Then, the right-hand side of this display simplifies to
\begin{align*}
    & \mathbb{P}\Big[W_{\cR, l}(s) \leq W_{\cR, l}(\sRl) - \left[\beta_{\sRl}(t, P_{\nu_{\cR}}^{\perp}y)\right]_s, \ \forall s \in \rchi^{*}_{\cR,l} \setminus \{\sRl\}\Big] \\
    = & \int_{\mathbb{R}} \mathbb{P}\Big[W_{\cR, l}(s) \leq w - \left[\beta_{\sRl}(t, P_{\nu_{\cR}}^{\perp}y)\right]_s, \forall s \in \rchi^{*}_{\cR,l} \setminus \{\sRl\} \mid W_{\cR, l}(\sRl) = w\Big] \phi(w;0,1)  dw \\
    = & \int_{\mathbb{R}} \displaystyle\prod_{s\in \rchi^{*}_{\cR,l}\setminus \{\sRl\}} \mathbb{P}\Big[W_{\cR, l}(s) \leq w - \left[\beta_{\sRl}(t, P_{\nu_{\cR}}^{\perp}y)\right]_s\Big] \phi(w;0,1) \ dw\tag{6} \label{eq:prod_decompose}\\
    = & \int_{\mathbb{R}} \displaystyle\prod_{s\in \rchi^{*}_{\cR,l}\setminus \{\sRl\}} \Phi\left(\dfrac{w-\left[\beta_{\sRl}(t, P_{\nu_{\cR}}^{\perp}y)\right]_s}{\tau_{P_{\cR,l}}}\right) \phi(w;0,1)\ dw \\
    =&\int_{\mathbb{R}} \Gamma_{\sRl}(w; t, P_{\nu_{\cR}}^{\perp}y)\phi(w;0,1)\ dw,
\end{align*}
where the decomposition in \eqref{eq:prod_decompose} follows from the independence of randomization variables across the splits, i.e.,
\begin{equation*}
    \begin{aligned}
 \mathbb{P}\Big[W_{\cR, l}(s) \leq w & - \left[\beta_{\sRl}(t, P_{\nu_{\cR}}^{\perp}y)\right]_s, \ \forall s \in \rchi^{*}_{\cR,l} \setminus \{\sRl\} \mid W_{\cR, l}(\sRl) = w\Big]\\
&= \mathbb{P}\Big[W_{\cR, l}(s) \leq w - \left[\beta_{\sRl}(t, P_{\nu_{\cR}}^{\perp}y)\right]_s, \ \forall s \in \rchi^{*}_{\cR,l} \setminus \{\sRl\}\Big], \; \text{ and}
   \end{aligned}
\end{equation*}
\begin{align*}
\mathbb{P}\Big[W_{\cR, l}(s) \leq w &- \left[\beta_{\sRl}(t, P_{\nu_{\cR}}^{\perp}y)\right]_s, \ \forall s \in \rchi^{*}_{\cR,l} \setminus \{\sRl\}\Big] \\
&= \displaystyle\prod_{s\in \rchi^{*}_{\cR,l}\setminus \{\sRl\}} \mathbb{P}\Big[W_{\cR, l}(s) \leq w - \left[\beta_{\sRl}(t, P_{\nu_{\cR}}^{\perp}y)\right]_s\Big].
\end{align*}
\end{proof}
 
As an immediate consequence of Proposition \ref{prop:1d-intergral}, it follows that the correction function in our pivot equals
$$
\Lambda_{\mathcal{R}}(t, P_{\nu_{\mathcal{R}}}^{\perp}y)= \prod_{l=1}^L \int_{\mathbb{R}} \Gamma_{\sRl}(w; t, P_{\nu_{\cR}}^{\perp}y)\phi(w;0,1) dw,
$$
i.e., it is a product of $L$ fairly simple $1$-dimensional integrals. 
This simplification leads to an efficient inference algorithm based on our exact pivot, where the univariate integral with respect to the standard normal density is numerically approximated via Gauss–Hermite quadrature, as outlined in Algorithm~\ref{algo:pivot}.

\begin{algorithm}[h]
\caption{Numerical algorithm for computing our pivot for $\nu_{\cR}^\top\mu$}
\label{algo:pivot}
\KwIn{Data $(X, Y)$, terminal region $\cR \in \bcR$, grid $G \gets \{t_0, t_1, ..., t_{M}: t_i - t_{i-1} = \delta\}$, $\delta > 0$, number of nodes for Gauss-Hermite quadrature $m \in \mathbb{N}$}
\KwOut{Numerically-evaluated pivot $\widetilde{P}(\nu_{\cR}^\top y,P_{\nu_{\cR}}^{\perp}y; \nu_{\cR}^\top \mu)$}

\ForEach{$t \in G$}{
    Compute and store the unadjusted density $\phi(t; \nu_{\cR}^\top \mu, \sigma^2||\nu_{\cR}||^2_2)$\;
    \ForEach{$l \in [L]$}{
        Compute $\beta_{\sRl} (t, P_{\nu_{\cR}}^{\perp}y)$ as in \eqref{eq:mean_opt_var}\;
        Numerically compute the split-wise selection probability 
        $$\widetilde\Lambda_{\sRl}(t, P_{\nu_{\cR}}^{\perp}y) \gets \text{Gauss\_Hermite}_m\left(\Gamma_{\sRl}(\cdot; t, P_{\nu_{\cR}}^{\perp}y)\right) \ ;$$
        Store $\widetilde\Lambda_{\sRl}(t, P_{\nu_{\cR}}^{\perp}y)$\;
    }
    Compute the numerical approximation to the selection probability 
        $$\widetilde\Lambda_{\mathcal{R}}(t, P_{\nu_{\mathcal{R}}}^{\perp}y) \gets \prod_{l=1}^L \widetilde\Lambda_{\sRl}(t, P_{\nu_{\cR}}^{\perp}y)\ ;$$
    Compute the numerical approximation to the conditional density 
        $$\widetilde{d}(t) \gets \phi(t; \nu_{\cR}^\top \mu, \sigma^2||\nu_{\cR}||^2_2) \widetilde\Lambda_{\mathcal{R}}(t, P_{\nu_{\mathcal{R}}}^{\perp}y)\ ;$$
}
Numerically compute the pivot as: 
$\widetilde{P}(\nu_{\cR}^\top y,P_{\nu_{\cR}}^{\perp}y; \nu_{\cR}^\top \mu)
\gets 
\left(\sum_{t\in G} \widetilde{d}(t)\right)^{-1}\cdot {\sum\limits_{t\in G: t \leq \nu_{\cR}^\top y} \widetilde{d}(t)}$\;

\Return $\widetilde{P}(\nu_{\cR}^\top y,P_{\nu_{\cR}}^{\perp}y; \nu_{\cR}^\top \mu).$
\end{algorithm}

\section{Empirical Analyses}
\label{sec:simulation}

\subsection{Simulations}
Here we present a simulation study to demonstrate the suitability of our proposed method and compare to existing approaches.
We generate data from a surface-plus-noise model, as described in Section \ref{sec:firstexample}, with $X \in \mathbb{R}^{n \times p}$ with $X_{i j} \stackrel{i.i.d.}{\sim} N(0,1)$, $n=200$, and $Y = \mu + \epsilon$ where random noise is drawn as $\epsilon_i \overset{i.i.d.}{\sim} D$.
The vector $\mu$ given by $$\mu_i=b \times\left[1_{\left(x_{i, 1} \leq 0\right)} \times\left\{1+a 1_{\left(x_{i, 2}>0\right)}+1_{\left(x_{i, 3} \times x_{i, 2}>0\right)}\right\}\right]$$ defines a three-level tree where $a=1$, $b=2$ determine the signal strength. 
We present results under Gaussian noise, where our pivot provides exact inference guarantees, while its performance under misspecified error distributions is deferred to Appendix~\ref{app:addsims}.

For each fitting method, we set the maximum depth of the final tree as 3, the minimum number of samples in a node to be split further as 50, the minimum size of terminal nodes as 20, and leave the grown tree unpruned after the stopping criterion is met.
After fitting the trees, we construct confidence intervals for parameters associated with the terminal regions, $\nu_{\cR}^\top \mu$ for $\cR \in \bcR$, using our proposed method along with the three baseline methods described in Section \ref{sec:firstexample}.
Note that $\nu_{\cR}^\top \mu$ is a well-defined parameter even when the fitted TREE-model only approximates the true generating distribution, which is never known in practice.
In any practical application, this parameter is always interpreted with respect to the model fitted to the observed data.

We consider the following evaluation metrics for comparison:
\begin{enumerate}
    \item \textbf{Coverage rate}: In each round of simulation, for a regression tree $TREE$ fitted with terminal nodes $\overline{R}$, we compute the coverage rate as ${\left|\left\{R \in \overline{R}: \nu_R^\top \mu \in \text{CI}_R\right\}\right|}/{|\overline{R}|}$.

    \item \textbf{Average CI length}: To measure the inferential power of the tests, we report the average length of the confidence intervals $\text{CI}_R = (L_R, U_R)$, i.e., $\{\sum_{R \in \overline{R}} (U_R - L_R)\}/{|\overline{R}|}$.
    \item \textbf{Test MSE}: To examine the model fitting quality of different methods, we generate new test data $Y^{\text{test}} = \mu + \epsilon^{\text{test}}$, where $\epsilon^{\text{test}}_i \overset{i.i.d.}{\sim} D$ is a new vector of independently drawn random noise. Then, for each sample we compute its predicted response value
    $\widehat{\mu}(X_i)=\sum_{R \in \overline{R}} \bar{Y}_R 1_{(X_i \in R)},$ and the test MSE, i.e., $ \sum_{i=1}^n (Y^{\text{test}}_{i} - \widehat{\mu}(X_i))^2/n$.
\end{enumerate}

\noindent \textbf{Results under varying Gaussian noise scales}. \ To evaluate the performance of the proposed method compared to the baseline methods under varying signal strengths, we fix $p=10$ and vary the noise distribution $D = \mathcal{N}(0, \sigma^2_G)$ for $\sigma_G \in \{1,2,5,10\}$ and employ our proposed method with the randomization sd parameter $\tau_P=\tau = c*\sigma_G$ for $c = 1$, i.e., RRT(1). We compare the empirical performance on the same simulated dataset over 500 simulations of our proposal with that of  the Tree-Values method and the UV method described in Section \ref{sec:intro}.
The resulting coverage rates, average confidence interval lengths, and test MSE are presented in Figure \ref{fig:vary_signal_G}. 

While all three methods approximately achieve the targeted coverage rate of 90\%, the RRT method produces confidence intervals that are shorter than Tree-Values intervals by orders of magnitude.
Additionally, our intervals are generally shorter than UV intervals across all settings, except in the case of the signal setting with the highest noise, i.e., $\sigma_G =10$.
Since there is almost no residual information in the data used for selection that could have been utilized for inference, our intervals are slightly longer than the UV intervals in this setting.
Furthermore, the proposed method results in favorable test MSE performance compared to the two baseline methods. See Appendix~\ref{app:addsims} for an additional  comparison that demonstrates these conclusions are similar under misspecification of the error distribution.

\begin{figure}[h!]
    \centering
    \includegraphics[width=\linewidth]{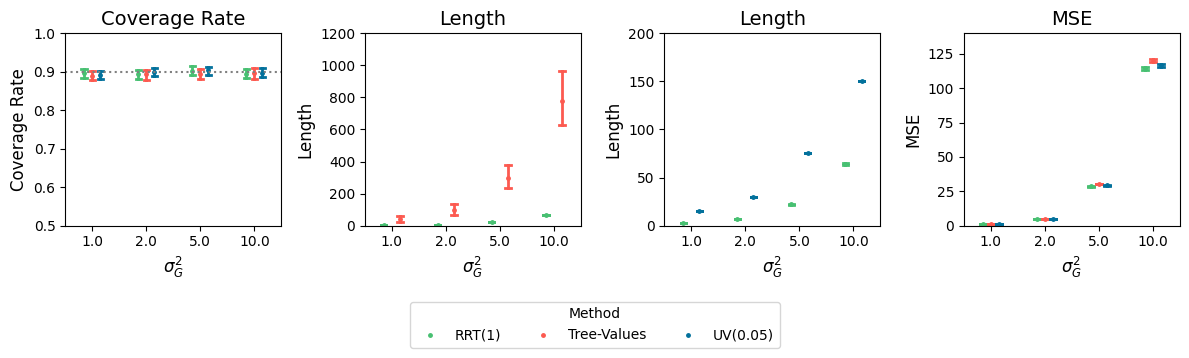}
    \captionsetup{width=\textwidth}
    \caption{Coverage rate, average CI length, and prediction MSE of RRT(1), Tree-Values, and the UV method for $\sigma^2_G \in \{1,2,5,10\}$ and Gaussian noise; {The dotted line is plotted at 0.9 in the coverage plot}}
    \label{fig:vary_signal_G}
    \vspace{-0.5cm}
\end{figure}

\noindent \textbf{Results under varying dimensions}. \
We generate data similar to the previous simulation. 
To evaluate the performance of the proposed method compared to the baseline methods under dense/sparse signals, we vary the number of covariates $p \in \{5, 10, 20\}$, corresponding to 2, 7, 17 noise variables, respectively.
\begin{figure}[h]
    \centering
    \includegraphics[width=\linewidth]{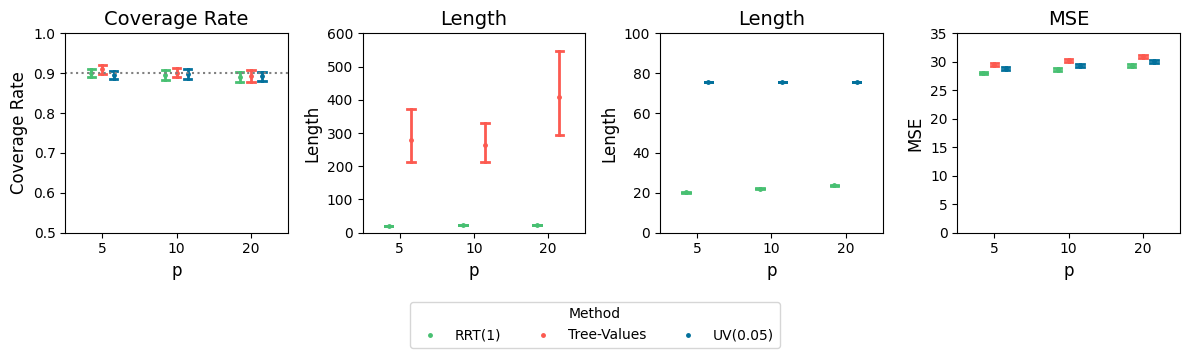}
    \captionsetup{width=\textwidth}
    \caption{Coverage rate, average CI length, and prediction MSE of RRT(1), UV method, and Tree-Values for $p \in \{5, 10, 20\}$; {The dotted line is plotted at 0.9 in the coverage plot}}
    \label{fig:vary_p}
    \vspace{-0.5cm}
\end{figure}
Similarly, in Figure \ref{fig:vary_p}, while all three methods approximately achieve the targeted coverage rate of 90\%, the proposed method yields intervals that are much shorter than Tree-Values intervals and the UV intervals on average.
An important takeaway from these results is that our method does not compromise predictive performance to inferential power, as reflected in the test MSE comparison.

Codes for reproducing the simulation results in our paper can be found at \url{https://github.com/yiling-h/SI-CART}.

\subsection{Case study: PROMPT}
\label{sec:case study} 
The \emph{PROviding Mental health Precision Treatment (PROMPT) Precision Health Study} is a 12-month mobile health intervention trial focused on reducing the burden of depression by augmenting standard mental health care using mobile health technologies to support patients on the wait list for traditional care. Adult patients (age 18+) who have a scheduled mental health intake appointments at either Michigan Medicine Outpatient Psychiatry or University Health Service clinics were eligible for participation. 
Recruited patients entered study at least 2 weeks prior to their initial clinic appointment.  
Participants were asked to complete surveys throughout the study, including an initial intake survey and a 6-week follow-up survey.
Each study participant received a Fitbit to wear daily for the duration of their time in the study. 


Here, we model the effects of the initial intake survey and wearable device data collected over the first study month on a measure of depression severity reported at the 6-week follow-up survey known as the Patient Health Questionnaire 9 (PHQ-9).
To predict the PHQ-9, we compute summary statistics, such as means and standard deviations, of 15 daily variables. Additionally, we aggregated individual patient responses to each of 9 different intake surveys, including the intake PHQ-9 and the intake General Anxiety Disorder (GAD-7) to create severity scores for each of these 9 surveys. After dropping variables with high missing rates (20\%), such as variables requiring consistent user input, we include 12 (summarized) sensor variables and 9 intake survey variables in the tree model. The final list of variables is included in Table \ref{table:var_names} from Appendix \ref{Appendix_PROMPT}. Our final dataset consists of $N=500$ patients with $21$ sensor and intake survey variables as predictors, and the PHQ-9 severity score as the response.

\noindent \textbf{Subsampling \& methods comparison}.
To compare the predictive and inferential power of the proposed RRT method with existing methods, we include Tree-Values and the UV method as baseline methods. For predictive power comparisons, we first split the dataset into a train set with $N_{\text{train}} = 300$ samples and a test set with $N_{\text{test}} = 200$ samples. Furthermore, for a comparison of all methods in a more realistic scientific setting with incoming data streams, we further subsample the train data with $50\%$ and $75\%$ samples of the full training set, with $N_{\text{train}50\%} = 150$ and $N_{\text{train}75\%} = 225$ samples, respectively.

All three methods are first fitted on the three datasets with $50\%$, $75\%$, and $100\%$ samples of the full data, with the maximum depth of the final tree as $4$, the minimum number of samples in a node to be split further as $50$, the minimum size of terminal nodes as $20$, and leave the grown tree unpruned after the stopping criterion is met. We set $\tau_P=\tau = c*\widehat\sigma$ for $c = 1$ for the RRT method, i.e., RRT($1$) and $\gamma = 0.05$ for the UV method, i.e., UV($0.05$), where $\widehat\sigma$ is the sample estimate for $\sigma$. The confidence intervals for the means of the terminal regions are computed after model fitting. Finally, the predictive performance of fitted models is evaluated on the holdout testing set, measured by the MSE.

\noindent \textbf{Empirical results and findings}.
Table \ref{table:CI_PROMPT} summarizes the average lengths of confidence intervals for the mean of terminal nodes given by the three different methods under different subsampling proportions of the training dataset. Consistent with the observations in the simulation study, the proposed method produces the shortest intervals on average (see lengths highlighted in bold). In particular, the Tree-Values method can produce much longer confidence intervals given the small sample sizes.

\begin{table}[H]
  \centering
  \begin{minipage}{0.4\textwidth}
    \centering
\begin{tabular}{@{}cccc@{}}
\toprule
Proportion  & 50\%    & 75\%   & 100\% \\ \midrule
Tree-Values & 12.588  & 80.500 & 20.474 \\
UV($0.05$)  & 10.745  & 12.510 & 12.017 \\
\textbf{RRT}($\bm{1}$)    & \textbf{6.822}   & \textbf{7.869} & \textbf{6.158} \\ \bottomrule
\end{tabular}
\caption{Average CI lengths for different proportions of training data}
\label{table:CI_PROMPT}
\end{minipage}
  \hfill
\begin{minipage}{0.45\textwidth}
    \centering
\begin{tabular}{@{}cccc@{}}
\toprule
Proportion  & 50\%    & 75\%    & 100\%   \\ \midrule
Tree-Values & 29.812 & 29.204 & 29.649 \\
UV($0.05$)  & 30.685 & 30.545 & 28.872 \\
\textbf{RRT}($\bm{1}$)    & \textbf{28.809} & \textbf{28.159} & \textbf{27.543} \\ \bottomrule
\end{tabular}
\caption{Test MSE for different proportions of training data}
\label{table:MSE_PROMPT}
\end{minipage}
\end{table}

Second, Table \ref{table:MSE_PROMPT} shows all three methods have similar MSE when validated on the test set, with RRT($1$) yielding the lowest MSE, highlighted in bold. The improved predictive power observed with added noise may be attributed to a regularization effect introduced by external randomization in this case. Ultimately, our finding echoes the observations in the simulation study and confirms that the proposed randomized procedure does not compromise the predictive power for a more powerful inference.

Finally, the tree fitted on the full training set using RRT is provided in Figure \ref{fig:PROMPT_TREE}. The intervals shown in the terminal nodes are the confidence intervals for the corresponding node means. 

Existing literature uses the following criteria to impute depression severity levels based on PHQ scores: 
minimal (0–4), mild (5–9), moderate (10–14), moderately severe (15–19), and severe (20–27), \cite{kroenke2001phq}. Based on the confidence intervals generated for the terminal nodes, we estimated that patients with high intake PHQ scores (PHQ\_B $>$ 25) fall into the range of mild-severe depression levels (node 9). For patients whose intake PHQ scores are within the range of 18 to 25, those having higher basal metabolic rates (CaloriesBMR $\leq 1325.46$) are estimated to have moderate-severe depression, while those having lower basal metabolic rates are estimated to have only mild-moderately severe depression (node 7 \& 8). In contrast, for patients with lower intake PHQ scores (PHQ\_B $\leq$ 18), the fitted tree uses other intake survey scores to predict the 6-week PHQ score. For these patients, none of the corresponding terminal nodes (nodes 1-6) falls in the range of severe 6-week depression based on the confidence intervals generated for these nodes.

\section{Acknowledgements}
We acknowledge the PROviding Mental health Precision Treatment (PROMPT) Precision Health Study at the University of Michigan for granting us access to the PROMPT data. Details of the study can be found at 
\url{https://um-prompt.wixsite.com/prompt}.
We gratefully acknowledge the commentary and feedback of Anna Neufeld, when an early version of this work was presented at the International Seminar on Selective Inference.

\bibliography{bibliography}

\newpage

\phantomsection\label{supplementary-material}
\bigskip

\begin{center}

{\large\bf SUPPLEMENTARY MATERIAL}

\end{center}

The article is self-contained; the Supplementary Material offers extended technical details and empirical results for interested readers. The Supplementary Material includes: full proofs of the main results (\ref{app:proofs}); a simplified overview of conditional inference after a single split (\ref{app:onesplit}); inference procedures under adaptive stopping rules (\ref{app:adaptive}); additional simulations, figures, and examples (\ref{app:addsims}); and further details of the \textsc{PROMPT} study (\ref{Appendix_PROMPT}).

\setcounter{section}{0}
\renewcommand{\thesection}{Supplement~\Alph{section}}
\renewcommand{\thesubsection}{Supplement~\Alph{section}.\arabic{subsection}}

\section{Simplified Overview of Conditional Inference after Single Split}
\label{app:onesplit}
Short description of Supplement A.
Before developing our inferential framework for RRT, we first illustrate our main idea of inference using a simple one-depth tree. 
We motivate our novel randomization approach and demonstrate how it enables easily tractable conditional inference through this simple example.
We consider the TREE-model in \eqref{tree:model} with exactly two terminal regions, $\cR_{1} = \cR$ and $\cR_{2} = \cR^c$, which are the result of a single split generated by $\Tree(Y, \bW)$. Here, $\bW= \bW_{1}$ denotes the randomization variables added to the gain functions at this first split. 
As before, we consider inference for $\mu_{\cR}$, where $\cR \in \{\cR_{1}, \cR_{2}\}$. 

Fixing some more notations, let $\rchi_{\cR,1}= \rchi_{1}$ denote the set of possible splits for the parent region $P_{1}=\mathbb{R}^p$ and let $|\rchi_{{\cR,1}}|=d_1$.
Let $S^{*}_{\cR,1}= S^{*}_{\cR,1}(Y, \bW_1)$, a function of both $Y$ and $\bW_1$, denote the first random split selected from this set that resulted in the region $\cR$.
Say that we observe the event $\{S_{\cR, 1}^{*} = s^{*}_{\cR, 1}\}$.
In this case, a pivot is obtained from the  conditional distribution of
\begin{equation}
\nu_{\cR}^\top Y  \mid P_{\nu_{\cR}}^{\perp}Y = P_{\nu_{\cR}}^{\perp}y, \bS = \bs,
\label{target:cond:dist}
\end{equation} where the conditioning event is   $\left\{\bS= \bs\right\}= \left\{S^{*}_{\cR,1}=s^{*}_{\cR, 1}\right\}.$
In the rest of this section, we guide our readers through the main steps of characterizing this conditional distribution.

Observe that we can express
$Y = \dfrac{\nu}{||\nu||_{2}^{2}} (\nu^{\top}Y) + P_{\nu}^{\perp}Y,$
where $\nu^{\top}Y$ and $P_{\nu}^\perp Y$ are independent variables, i.e., $\nu^{\top}Y\indep P_{\nu}^\perp Y$.
Let $\nu = \nu_{\cR}$, where $\nu_{\cR}$ is defined in the previous section.
As a result, the random split $S_{\cR, 1}^{*}$, made on the parent region $P_{1}$ is fully determined by the data variables $\nu_{\cR}^{\top}Y$, $P_{\nu_{\cR}}^{\perp}Y$, and the randomization variables $\bW_{1}$, i.e.,
$S_{\cR, 1}^{*} = S_{\cR, 1}^{*}(\nu_{\cR}^{\top}Y, P_{\nu_{\cR}}^{\perp}Y, \bW_{1})$.

With details deferred to the next section, it follows that the conditional distribution of interest has a density at
$t \in \mathbb{R}$ proportional to:
$$
\phi \nbracket{t; \nu_{\cR}^\top\mu, \sigma^{2}||\nu_{\cR}||^2_2 }\times  \mathbb{P}\rbracket{S_{\cR, 1}^{*}  = s^{*}_{\cR, 1} \mid P_{\nu_{\cR}}^{\perp}Y = P_{\nu_{\cR}}^{\perp}y, \nu_{\cR}^\top Y  = t},
$$
where $\nu_{\cR}^\top\mu =\sqrt{n_{\cR}}\mu_{\cR}$ and $\phi(t; \mu, \gamma^2)$ denotes a normal density with mean parameter $\mu$ and variance $\gamma^2$.
In the above-stated conditional density, the second term represents the probability of observing the split $s^{*}_{\cR, 1}$, conditional on the data variables. 
When combined with the na{\"i}ve density of $\nu_{\cR}^\top Y$, this yields the conditional density that provides valid inference in the RRT.
Computing this probability is central to the new inference approach, and due to the additive form of our randomization, we can easily derive an exact expression for it.

Below, we outline the main two steps for computing this probability, providing a   formal derivation of its expression in the next section.

\noindent \textbf{Step 1}.
First, observe that our event equals
\begin{equation*}
 \begin{aligned}
    \cbracket{S^{*}_{\cR,1}= s^{*}_{\cR,1}} 
     =& \Big\{W_{1}(s^{*}_{\cR,1}) - W_{1}(s) \geq  G(Y;P_{1},s) - G(Y;P_{1},  s^{*}_{\cR,1}),\ \forall s \in \rchi_{\cR,1}\setminus \{s^{*}_{\cR,1}\}\Big\}.
 \end{aligned}
 \end{equation*} 

\noindent \textbf{Step 2}.
Letting $y(t)=  t\dfrac{\nu_{\cR}}{||\nu_{\cR}||_{2}^{2}} + P_{\nu_{\cR}}^{\perp}y$ and taking probabilities of the two equivalent events in Step 1, we have
 \begin{align*}
& \mathbb{P}\rbracket{S^{*}_{\cR,1}= s^{*}_{\cR,1} \mid \nu_{\cR}^{\top}Y =t, P_{\nu_{\cR}}^{\perp}Y=P_{\nu_{\cR}}^{\perp}y}\\
     =&  \mathbb{P}\Big[W_{1}(s^{*}_{\cR,1}) - W_{1}(s) \geq  G(Y;P_{1},s) - G(Y;P_{1},  s^{*}_{\cR,1}), \\
     & \;\;\;\;\;\;\;\;\;\;\;\;\;\;\;\;\;\;\;\;\;\;\;\;\;\;\;\;\;\;\;\;\;\;\;\;\;\;\;\;\;\;\;\;\;\;\;\;\;\;\;\;\;\;\;\;\;\;\;\;\; \forall s \in \rchi_{\cR,1} \setminus \{s^{*}_{\cR,1}\} \mid \nu_{\cR}^{\top}Y =t, P_{\nu_{\cR}}^{\perp}Y=P_{\nu_{\cR}}^{\perp}y \Big] \\
      =&\mathbb{P}\Big[W_{1}(s^{*}_{\cR,1}) - W_{1}(s) \geq  G(y(t);P_{1},s) - G(y(t);P_{1},  s^{*}_{\cR,1}),  \forall s \in \rchi_{\cR,1} \setminus \{s^{*}_{\cR,1}\} \Big].
 \end{align*}
The second equality uses the independence of the randomization variables from $Y$, and consequently, their independence from both $\nu_{\cR}^{\top}Y$ and  $P_{\nu_{\cR}}^{\perp}Y$.  Therefore, as long as we know the distribution of the differences $\{W_{1}(s^{*}_{\cR,1}) - W_{1}(s): \ s\in \rchi_{\cR,1} \setminus \{s^{*}_{\cR,1}\}\}$, this probability is straightforward to compute. 
This is indeed the case, as we demonstrate in the next section. Due to the normal distribution of the randomization variables, this probability simplifies to an integral based on the Gaussian density of these randomization variables.

\section{Proofs of main results}
\label{app:proofs}
\begin{proof}[Proof of Lemma \ref{L1}]
 Observe that
 \begin{equation*}
 \begin{aligned}
 & \mathbb{P}\rbracket{\bS = \bs \mid P_{\nu_{\cR}}^{\perp}Y =   P_{\nu_{\cR}}^{\perp}y, \nu_{\cR}^\top Y  = t }   \\
  =&\mathbb{P}\rbracket{\SR = \sR, \SRc=\sRc \mid P_{\nu_{\cR}}^{\perp}Y =   P_{\nu_{\cR}}^{\perp}y, \nu_{\cR}^\top Y  = t }   \\
   =&\underbrace{\mathbb{P}\rbracket{\SR = \sR \mid \SRc=\sRc , P_{\nu_{\cR}}^{\perp}Y =   P_{\nu_{\cR}}^{\perp}y, \nu_{\cR}^\top Y  = t }}_{\text{(P1)}}   \underbrace{\mathbb{P}\rbracket{\SRc=\sRc \mid P_{\nu_{\cR}}^{\perp}Y =   P_{\nu_{\cR}}^{\perp}y, \nu_{\cR}^\top Y  = t }}_{\text{(P2)}}.
  \end{aligned}
  \end{equation*}
  To prove our claim, we show the following
  \begin{align*}
  \begin{gathered}
  \text{(P1)}= \mathbb{P}\rbracket{\SR = \sR \mid P_{\nu_{\cR}}^{\perp}Y =   P_{\nu_{\cR}}^{\perp}y, \nu_{\cR}^\top Y  = t},  \\
   \text{(P2)}=  \mathbb{P}\rbracket{\SRc=\sRc \mid P_{\nu_{\cR}}^{\perp}Y =   P_{\nu_{\cR}}^{\perp}y}.
   \end{gathered}
  \end{align*}
This leads us to conclude that: 
 \begin{equation*}
 \begin{aligned}
 & \mathbb{P}\rbracket{\bS = \bs \mid P_{\nu_{\cR}}^{\perp}Y =   P_{\nu_{\cR}}^{\perp}y, \nu_{\cR}^\top Y  = t }  \propto  \mathbb{P}\rbracket{\SR = \sR \mid P_{\nu_{\cR}}^{\perp}Y =   P_{\nu_{\cR}}^{\perp}y, \nu_{\cR}^\top Y  = t}.
  \end{aligned}
  \end{equation*}

Firstly, note that any split $S^{*}_{j} \in \bS$ made on $P_{j} \in \bP$ is a function of $(\{Y_{i}: X_{i} \in P_{j}\}, \bW_{j})$, which are the observations whose covariates fall within $P_j$ and the independent randomization variable that was added to the gain functions while selecting this split.
Define 
$$Y(t)=  t\dfrac{\nu}{||\nu||_{2}^{2}} + P_{\nu_{\cR}}^{\perp}Y.$$  
Then for a pair of splits $(S^{*}_{j},S^{*}_{j'}) \in \bS \times \bS$, such that $j\neq j'$, we have that 
$S^{*}_{j} \indep S^{*}_{j'} \mid Y(t) = y(t)$, 
since $\bW_{j} \indep \bW_{j'}$ for splits made on distinct parent regions. 
As a result, we also have that
 $$S^{*}_{j} \indep S^{*}_{j'} \mid \cbracket{P_{\nu_{\cR}}^{\perp}Y =   P_{\nu_{\cR}}^{\perp}y, \nu_{\cR}^\top Y  = t},$$
and that $\SR \indep \SRc | \cbracket{P_{\nu_{\cR}}^{\perp}Y =   P_{\nu_{\cR}}^{\perp}y, \nu_{\cR}^\top Y  = t}$.
This proves the claim about $\text{(P1)}$.

\bigskip
Now, for any $S^{*}_{j} \in \SRc$ made on $P_{j} \in \bP$, we have that $\cR \cap P_{j} = \emptyset$.
For such a $P_{j}$ with a disjoint intersection with $\cR$, given that $\bW \indep Y$ and that $\{Y_{i}:  X_{i}\in P_j\} \indep \nu_{\cR}^\top Y$,  we have
$$S^{*}_{j} \indep \nu_{\cR}^\top Y.$$
This proves our claim about $\text{(P2)}$ immediately.
 \end{proof}



\begin{proof}[Proof of Proposition \ref{prop:1}]
The conditional density of $\nu_{\cR}^\top Y \mid \cbracket{\bS = \bs, P_{\nu_{\cR}}^{\perp}Y = P_{\nu_{\cR}}^{\perp}y}$ at $t$ is proportional to
 $$
 \phi \nbracket{t; \nu_{\cR}^\top\mu, \sigma^{2}||\nu_{\cR}||^2_2 }\times  \mathbb{P}\rbracket{\bS(Y, \bW) = \bs \mid \nu_{\cR}^\top Y  = t, P_{\nu_{\cR}}^{\perp}Y =   P_{\nu_{\cR}}^{\perp}y},
 $$
due to the independence between $\nu_{\cR}^\top Y$ and $P_{\nu_{\cR}}^{\perp}Y$.
Using the conclusion in Lemma \ref{L1}, we futher note that this density is proportional to 
$$
\phi \nbracket{t; \nu_{\cR}^\top\mu, \sigma^{2}||\nu_{\cR}||^2_2 } \times  \mathbb{P}\rbracket{\SR = \sR \mid P_{\nu_{\cR}}^{\perp}Y =   P_{\nu_{\cR}}^{\perp}y, \nu_{\cR}^\top Y  = t},
$$
since 
$$\mathbb{P}\rbracket{\bS = \bs \mid \nu_{\cR}^\top Y  = t, P_{\nu_{\cR}}^{\perp}Y =   P_{\nu_{\cR}}^{\perp}y} 
     \propto  \mathbb{P}\rbracket{\SR = \sR \mid \nu_{\cR}^\top Y  = t, P_{\nu_{\cR}}^{\perp}Y =   P_{\nu_{\cR}}^{\perp}y}.$$
This leads to the claimed conditional density.
 \end{proof}

\begin{proof}[Proof of Theorem  \ref{pivot:1}]
Given the conditional density of $\nu_{\cR}^\top Y \mid \cbracket{\bS = \bs, P_{\nu_{\cR}}^{\perp}Y = P_{\nu_{\cR}}^{\perp}y}$ in Proposition \ref{prop:cond:density}, we apply the probability integral transform (PIT) to obtain the pivot
$$P(\nu_{\cR}^\top y,P_{\nu_{\cR}}^{\perp}y; \nu_{\cR}^\top \mu) = \frac{\int_{-\infty}^{\nu_{R}^\top y} \phi(t; \nu_{\cR}^\top \mu, \sigma^2||\nu_{\cR}||^2_2) \Lambda_{\sR} (t, P_{\nu_{\cR}}^{\perp}y) dt}{\int_{-\infty}^{+\infty}\phi(t; \nu_{\cR}^\top \mu, \sigma^2||\nu_{\cR}||^2_2) \Lambda_{\sR} (t, P_{\nu_{\cR}}^{\perp}y)dt}.$$
Furthermore, because of its construct using the PIT, it follows directly that the pivot is distributed as a $\text{Unif}(0,1)$ random variable.
 \end{proof}

\section{Pivot under adaptive stopping rules}
\label{app:adaptive}
Although the primary inferential results in Section \ref{Sec:mainresults} were developed for fixed-depth trees or more generally, trees grown using deterministic stopping rules, we demonstrate here that our RRT method  can be extended to accommodate other TREE-models grown with adaptive stopping rules. 
This is achieved by making slight modifications to the standard rules with external randomization. 
Specifically, we discuss randomized variants of two commonly used adaptive stopping rules. 
For both types of adaptively grown TREE-models, we show how inference for the data-dependent model parameters can be conducted using a similar approach as developed for the fixed-depth trees in Section \ref{Sec:mainresults}.

\noindent\textbf{Rule based on thresholding the gain function}.
Consider growing the classic regression tree, $TREE^{\lambda}(Y;X)$, without external randomization, in the following manner.
At a region $P_k\subset \mathbb{R}^{p}$, we first select the split $s_k^{*} \in \rchi_{P_k}$ that maximizes the gain function $G(Y; P_k, s)$; but, the split is made only if the gain $G(Y; P_k, s_k^{*})$ exceeds a pre-specified threshold $\lambda$; otherwise, we stop.

To grow $TREE^{\lambda}(Y, \bW;X)$ with added randomization variables, we apply the same rule, but this time to the randomized gain functions rather than the standard gain function.
More precisely, at the parent region $P_k\subset \mathbb{R}^{p}$, we observe 
$
\cbracket{S^{*}_{k}\nbracket{Y, \bW_{k}; P_{k}} = s^{*}_{k}} 
$
if and only if 
\begin{equation}
\label{sel:event:ASR}
\left\{s^{*}_{k} = \underset{s\in \rchi_{k}}{\argmax}\; G(Y; P_{k}, s)  + W_k(s)\right\} \cap \Big\{G(Y; P_{k}, s^{*}_{k})  +W_{k}(s^{*}_{k})\geq \lambda\Big\}
\end{equation}
where $\bW_{k} = \Big\{W_{k}(s),\; s\in \chi_k\Big\}$ are randomization variables drawn independently from $\mathcal{N}(0, \tau_{P_k}^2 )$ and also independent of $Y$.
Algorithm \ref{alg:RRTa1} summarizes the steps to build an adaptively grown TREE-model using this randomized thresholding rule. 

To address inference with this adaptive stopping rule, note that the event in \eqref{sel:event:ASR}
is equivalent to
\begin{align*}
 \Big\{W_{k}(s^{*}_{k}) - W_{k}(s)  & \geq  G(Y;P_{k},s) - G(Y;P_{k},  s^{*}_{k}),\ \forall s \in \rchi_{k}\setminus \{s^{*}_{\cR,1}\}\Big\} \\
&\;\;\;\;\;\;\;\;\;\;\;\;\;\;\;\;\;\;\;\;\;\;\;\;\;\;\;\;\;\;\;\;\;\;\;\;\;\;\;\;\;\;\cap \Big\{ W_{k}(s^{*}_{k}) \geq \lambda - G(Y; P_{k}, s^{*}_{k})\Big\}.
\end{align*}

For the series of $L$ selected splits $\sR=\cbracket{s^{*}_{\cR,1}, s^{*}_{\cR,2}, \ldots, s^{*}_{\cR,L}}$ in the sub-tree leading to the terminal region $\cR$, define for each $l \in [L]$:
\begin{align*}
O^{(1)}_{\sRl}(t, P_{\nu_{\cR}}^{\perp}y)&=\nbracket{G\nbracket{y(t) ;P_{\cR,l}, \sRl}-\lambda, \beta_{\sRl}(t, P_{\nu_{\cR}}^{\perp}y)}^{\top} 
\end{align*}
where $y(t)$ and $\beta_{\sRl}(t, P_{\nu_{\cR}}^{\perp}y)$ are as defined before.

Then, for $\widetilde{\Omega}^{(1)}_{\cR, l} = \begin{bmatrix}
   \tau^2_{P_{\cR,l}} & \tau^2_{P_{\cR,l}} \cdot 1_{d_{l-1}}^{\top}\\
   \tau^2_{P_{\cR,l}} \cdot 1_{d_{l-1}} & \Omega_{\cR, l}
\end{bmatrix}\in \mathbb{R}^{d_{l}\times d_{l}}$, we define
\begin{align*}
\widetilde{\Lambda}^{(1)}_{\sRl}(t, P_{\nu_{\cR}}^{\perp}y) &= \int_{\mathbb{R}_{+}^{d_{l}}} \phi\nbracket{u; O^{(1)}_{\sRl}(t, P_{\nu_{\cR}}^{\perp}y), \widetilde{\Omega}^{(1)}_{d_l}} du. 
\end{align*}

In Theorem \ref{thm:pivot:3}, we present our pivot for inference in the TREE-model, derived from the output of Algorithm \ref{alg:RRTa1}.
\begin{theorem}
\label{thm:pivot:3}
Given data $Y=y$, a pivot for $\nu_{\cR}^\top \mu$ equals $$P^{(1)}(\nu_{\cR}^\top y,P_{\nu_{\cR}}^{\perp}y; \nu_{\cR}^\top \mu) = \frac{\int_{-\infty}^{\nu_{R}^\top y} \phi(t; \nu_{\cR}^\top \mu, \sigma^2||\nu_{\cR}||^2_2) \widetilde{\Lambda}^{(1)}_{\sR} (t, P_{\nu_{\cR}}^{\perp}y) dt}{\int_{-\infty}^{+\infty}\phi(t; \nu_{\cR}^\top \mu, \sigma^2 ||\nu_{\cR}||^2_2) \widetilde{\Lambda}^{(1)}_{\sR} (t, P_{\nu_{\cR}}^{\perp}y)dt},$$
which, conditional on $\{\bS=\bs\}$, is distributed as a $\text{\normalfont Unif}(0,1)$ random variable.   
\end{theorem}

The of Theorem  \ref{thm:pivot:3} closely resembles the proof for deriving the pivot related to the adaptive stopping rule based on cost complexity, which we will discuss next. Therefore, we will provide the proof for both of these pivots in the following section.


\begin{algorithm}[h]
\caption{Randomized CART: tree growing algorithm with adaptive stopping (1)}
\label{alg:RRTa1}
\KwIn{Training data $(X, Y)$, maximum depth $d_{\max}$, pruning parameter $\lambda$, noise variance $\tau_P^2$}
\KwOut{Set of terminal regions $\{\cR_1, \cR_2, \ldots, \cR_K\}$ forming the fitted adaptive randomized tree}

Initialize root region $P_1 \gets \mathbb{R}^p$, current depth $d \gets 0$\;

\While{$d < d_{\max}$}{
    \ForEach{region $P$ at depth $d$ with $n_P > 1$}{
        Compute impurity gain $G(Y; P, s)$ for all $s = (j, o) \in \chi_P$\;
        Draw random perturbations $W(s) \sim \mathcal{N}(0, \tau_P^2)$ for all $s \in \chi_P$\;
        Select randomized best split 
        $$s^* \gets \arg\max_{s \in \chi_P} \{ G(Y; P, s) + W(s) \}\ ;$$
        
        \eIf{$G(Y; P, s^*) + W(s^*) \ge \lambda$}{
            Partition region $P$ using $s^*$ into left and right children $\{P^{l}_{s^*}, P^{g}_{s^*}\}$\;
            Increment tree depth: $d \gets d + 1$\;
        }{
            Return $P$ as a terminal region\;
        }
    }
}
Return all terminal regions $\{P\}$ as $\{\cR_1, \ldots, \cR_K\}$\;
\end{algorithm}

\noindent\textbf{Rule based on a cost complexity criterion}.
Another commonly used adaptive stopping rule is based on a cost-complexity strategy, which in the standard CART is also employed for bottom-up pruning. 
This approach removes all descendants of a region $P$ if the gain in the sum of squared errors based on the terminal regions of $P$  is less than a pre-specified, cost-complexity threshold $\lambda$.
That is, let $\operatorname{TERM}(P)$ be the set of all terminal nodes with $P$ as parent.
The decision to not further split region $P$ is made if the average gain in MSE, defined as 
$$\text{GM}(Y; P) = \frac{\sum_{i: X_i \in P}\left(Y_i-\overline{Y}_P\right)^2-\sum_{R \in \operatorname{TERM}(P)} \sum_{i: X_i \in R}\left(Y_i-\overline{Y}_R\right)^2}{|\operatorname{TERM}(P)|-1}$$
does not exceed the complexity threshold $\lambda$.
The RRT can incorporate a similar adaptive stopping rule with external randomization.

Similar to the previous rule, we modify the splitting criterion by using a randomized version of the function $\text{GM}(Y; P)$.
Consider the parent region $P_k\subset \mathbb{R}^{p}$ in our RRT. 
First, we apply the standard CART algorithm with a fixed depth $d_0$, using $P_k$ as the root node, and compute $\text{GM}(Y; P_k)$, based on the terminal regions of the CART output.
Then, in our RRT, we observe 
$
\cbracket{S^{*}_{k}\nbracket{Y, \bW_{k}; P_{k}} = s^{*}_{k}} 
$
if and only if 
\begin{equation}
\label{sel:event:CCP}
\Big\{s^{*}_{k} = \underset{s\in \rchi_{k}}{\argmax} \; G(Y; P_{k}, s)  + W_{k}(s)\Big\} \cap \cbracket{\text{GM}(Y; P_k)  + \widetilde{W}_k \geq \lambda},
\end{equation}
where
$\bW_{k} = \Big\{W_{k}(s),\; s\in \chi_k\Big\} \cup  \left\{\widetilde{W}_k\right\}$ 
are $d_{k}+1$ data independent randomization terms drawn from $\mathcal{N}(0, \tau_{P_k}^2 )$. 
This adaptive variant of the RRT is outlined in Algorithm \ref{alg:RRTa2}.
Obviously, the event in \eqref{sel:event:CCP} is equivalent to 
\begin{align*}
\Big\{W_{k}(s^{*}_{k}) - W_{k}(s) \geq  & G(Y;P_{k},s) - G(Y;P_{k},  s^{*}_{k}),\quad \forall s \in \rchi_{k}\setminus \{s^{*}_{k}\}\Big\} \\
&\;\;\;\;\;\; \cap \cbracket{\text{GM}(Y; P_k)  + \widetilde{W}_{k} \geq \lambda}.
\end{align*}

\begin{algorithm}[h]
\caption{Randomized CART: tree growing algorithm with adaptive stopping (2)}
\label{alg:RRTa2}
\KwIn{Training data $(X, Y)$, maximum depth $d_{\max}$, pruning parameter $\lambda$, noise variance $\tau_P^2$}
\KwOut{Set of terminal regions $\{\cR_1, \cR_2, \ldots, \cR_K\}$ forming the fitted adaptive randomized tree}

Initialize root region $P_1 \gets \mathbb{R}^p$, current depth $d \gets 0$\;

\While{$d < d_{\max}$}{
    \ForEach{region $P$ at depth $d$ with $n_P > 1$}{
        Draw independent random perturbations 
        $$\{W(s): s \in \chi_P\}, \ \widetilde{W} \sim \mathcal{N}(0, \tau_P^2)\ ;$$
        Compute impurity gain $G(Y; P, s)$ for all $s = (j, o) \in \chi_P$\;
        Select randomized best split 
        $$s^* \gets \arg\max_{s \in \chi_P} \{ G(Y; P, s) + W(s) \}\ ;$$
        Grow a full CART subtree of depth $d_{\max}$ from $P$ to obtain $\operatorname{TERM}(P)$\;
        Compute the growth metric $\mathrm{GM}(Y; P)$\;
        
        \eIf{$\mathrm{GM}(Y; P) + \widetilde{W} \ge \lambda$}{
            Partition region $P$ using $s^*$ into left and right children $\{P^{l}_{s^*}, P^{g}_{s^*}\}$\;
            Increment tree depth: $d \gets d + 1$\;
        }{
            Return $P$ as a terminal region\;
        }
    }
}
Return all terminal regions $\{P\}$ as $\{\cR_1, \ldots, \cR_K\}$\;
\end{algorithm}

Consider observing $\sR=\cbracket{s^{*}_{\cR,1}, s^{*}_{\cR,2}, \ldots, s^{*}_{\cR,L}}$.
Using the notations introduced earlier, define for each $l \in [L]$:
\begin{align*}
O_{\sRl}^{(2)}(t, P_{\nu_{\cR}}^{\perp}y)&=\nbracket{\text{GM}(Y; P_{\cR, l})-\lambda, \beta_{\sRl}(t, P_{\nu_{\cR}}^{\perp}y)}^{\top}.
\end{align*}
For 
$\widetilde{\Omega}_{\cR, l}^{(2)}= \begin{bmatrix}
   \tau_{P_{\cR, l}}^2  & 0_{d_{l-1}}^{\top}\\
   0_{d_{l-1}} & \Omega_{\cR, l}
\end{bmatrix} \in \mathbb{R}^{d_{l}\times d_{l}}
$, let
\begin{align*}
\widetilde{\Lambda}^{(2)}_{\sRl}(t, P_{\nu_{\cR}}^{\perp}y) &= \int_{\mathbb{R}_{+}^{d_{l}}} \phi\nbracket{u; O^{(2)}_{\sRl}(t, P_{\nu_{\cR}}^{\perp}y), \Omega^{(2)}_{\cR, l}} du.
\end{align*} 

In Theorem \ref{thm:pivot:4}, we provide a pivot for inference in the TREE-model, based on the output of Algorithm \ref{alg:RRTa2}.
\begin{theorem}
\label{thm:pivot:4}
Given data $Y=y$, a pivot for $\nu_{\cR}^\top \mu$ equals $$P^{(2)}(\nu_{\cR}^\top y,P_{\nu_{\cR}}^{\perp}y; \nu_{\cR}^\top \mu) = \dfrac{\int_{-\infty}^{\nu_{R}^\top y} \phi(t; \nu_{\cR}^\top \mu, \sigma^2||\nu_{\cR}||^2_2) \widetilde{\Lambda}^{(2)}_{\sR} (t, P_{\nu_{\cR}}^{\perp}y) dt}{\int_{-\infty}^{+\infty}\phi(t; \nu_{\cR}^\top \mu, \sigma^2 ||\nu_{\cR}||^2_2) \widetilde{\Lambda}^{(2)}_{\sR} (t, P_{\nu_{\cR}}^{\perp}y)dt},$$
which, conditional on $\cbracket{\bS=\bs}$, is distributed as a $\text{\normalfont Unif}(0,1)$ random variable.   
\end{theorem}

\begin{namedproof}{Proofs of Theorem \ref{thm:pivot:3} and \ref{thm:pivot:4}}

Using the same strategy as adopted for the proof of 
Theorem \ref{pivot:1}, 
the conditional density of $\nu_{\cR}^\top Y \mid \cbracket{\bS = \bs, P_{\nu_{\cR}}^{\perp}Y = P_{\nu_{\cR}}^{\perp}y}$, when evaluated at $t\in \mathbb{R}$ is proportional to 
$$\phi \nbracket{t; \mu_{\cR}, \sigma^{2}||\nu_{\cR}||^2_2 } \times  \mathbb{P}\rbracket{\SR = \sR \mid P_{\nu_{\cR}}^{\perp}Y =   P_{\nu_{\cR}}^{\perp}y, \nu_{\cR}^\top Y  = t}.$$
Now it remains to compute 
$$
\mathbb{P}\rbracket{\SR = \sR \mid P_{\nu_{\cR}}^{\perp}Y =   P_{\nu_{\cR}}^{\perp}y, \nu_{\cR}^\top Y  = t}.
$$
To do so, define
$$
Z_{\cR,l}^{(1)}= \nbracket{G(Y; P_{\cR, l}, s^{*}_{\cR, l})+ W_{\cR, l}(s^*_{\cR,l})-\lambda, D_{\cR,l} }^{\top}; \; Z_{\cR, l}^{(2)}= \nbracket{\text{GM}(Y; P)  +\widetilde{W}_{\cR, l} -\lambda, D_{\cR,l} }^{\top}
$$
The proofs of Theorem \ref{thm:pivot:3} and \ref{thm:pivot:4} follow the same steps as before, as shown in Proposition \ref{prop:cond:density}, utilizing the fact that
\begin{align*}
\begin{gathered}
Z_{\cR,l}^{(1)} \; \lvert \; Y=y(t) \sim \mathcal{N}\nbracket{O^{(1)}_{\sRl}(t, P_{\nu_{\cR}}^{\perp}y), \widetilde{\Omega}_{\cR, l}}, \;\;
Z_{\cR,l}^{(2)} \; \lvert \; Y=y(t) \sim \mathcal{N}\nbracket{O^{(2)}_{\sRl}(t, P_{\nu_{\cR}}^{\perp}y), \widetilde{\Omega}_{\cR, l}}
\end{gathered}
\end{align*}
and  
$$\left\{\SR = \sR\right\} =\left\{Z_{\cR,l}^{(k)} >0 \ \text{ for } l\in [L]\right\}$$
by setting $k=1,2$ respectively for the two adaptive rules.  
\end{namedproof}

\section{Additional simulations}
\label{app:addsims}

\noindent \textbf{Results under varying Laplace noise scales}
In this experiment, we investigate the robustness of the proposed method, as well as the baseline methods, to misspecification of the noise distribution.
We set $D = \text{Laplace}(0, \sigma_L/\sqrt{2})$, for $\sigma_L \in \{1,2,5,10\}$, such that $\text{SD}(\epsilon_i) \in \{1, 2, 5, 10\}$. We set the dimension $p=10$. 

Following this, we compute selective inference using our proposed method and the other two baseline methods and report the results obtained from $500$ simulations.
The resulting coverage rates, average confidence interval lengths, and test MSE are presented in Figure \ref{fig:vary_signal_L}.
\begin{figure}[h!]
    \centering
    \includegraphics[width=\linewidth]{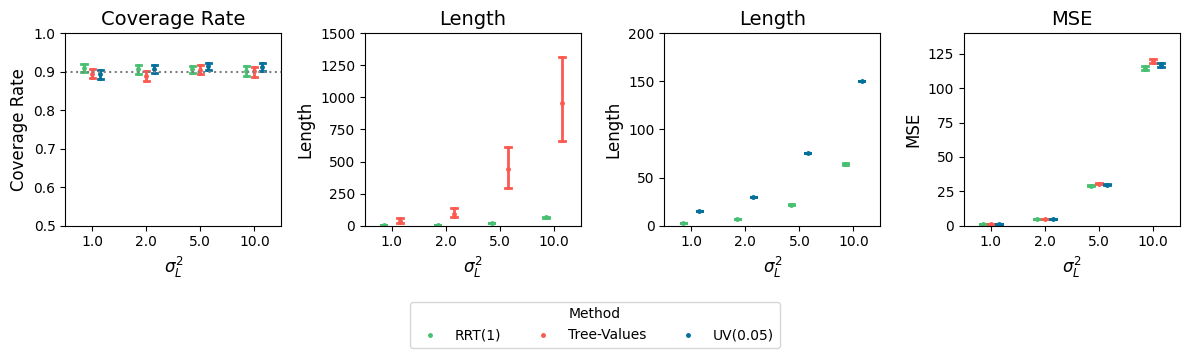}
    \caption{Coverage rate, average CI length, and prediction MSE of RRT(1), Tree-Values, and the UV method for $\sigma^2_L \in \{1,2,5,10\}$ and Laplace noise; {The dotted line is plotted at 0.9 in the coverage plot}}
    \label{fig:vary_signal_L}
\end{figure}

We note that similar trends are observed, even under a misspecified noise distribution: all three methods approximately achieve the targeted coverage rate of 90\%, and the proposed method produces confidence intervals that are shorter than Tree-values intervals by orders of magnitude, and shorter than or comparable to UV intervals in terms of lengths, with the proposed method having favorable test MSE performance compared to the two baseline methods. 

\section{Additional details for the PROMPT study}
\label{Appendix_PROMPT}
\begin{table}[H]
\scriptsize
\centering
\begin{tabular}{@{}cc@{}}
\toprule
 & Description                                                                                  \\ \midrule
ActivityCalories           & Calories burned from periods above sedentary level, personal average                         \\
BodyBmi                    & Body Mass Index, from the Body Time Series, personal average                                               \\
BodyWeight                 & Body weight, from the Body Time Series, personal average                                     \\
Calories                   & Calories, from the Activity Time Series, personal average                                    \\
CaloriesBMR                & Only BMR (Basal Metabolic Rate) calories, from the Activity Time Series, personal average    \\
Distance                   & Distance traveled, from the Activity Time Series, personal average                           \\
HeartRateIntradayCount     & The number of intraday heart rate samples collected during the time period, personal average \\
ActivityCaloriesSD         & Calories burned from periods above sedentary level, personal sd                              \\
CaloriesSD                 & Calories, from the Activity Time Series, personal average, personal sd                       \\
CaloriesBMR\_SD            & Only BMR (Basal Metabolic Rate) calories, from the Activity Time Series, personal sd         \\
DistanceSD                 & Distance traveled, from the Activity Time Series, personal average, personal sd              \\
HeartRateIntradayCountSD   & The number of intraday heart rate samples collected during the time period, personal sd      \\
ASSIST\_B                     & The Alcohol, Smoking and Substance Involvement Screening Test score at baseline (intake) survey                                                                                 \\
GAD\_B                        &  General Anxiety Disorder Survey score at baseline (intake) survey  \\
ISEL\_B                       & Interpersonal Support Evaluation List score at baseline (intake) survey \\
NEO\_B                        & NEO Personality Inventory score at baseline (intake) survey \\
PANSI\_B                      & Positive and Negative Suicide Ideation score at baseline (intake) survey \\
PCL\_B                        & PTSD Checklist score at baseline (intake) survey \\
PHQ\_B                        & Patient Health Questionnaire score at baseline (intake) survey \\
PSQI\_B                       & Pittsburgh Sleep Quality Index score at baseline (intake) survey \\
RFQ\_B                        & The Reflective Functioning Questionnaire score at baseline (intake) survey \\                
PHQ (response variable)                        & Patient Health Questionnaire score at 6-week survey \\ \bottomrule
\end{tabular}
\caption{Documentation for variables included in the PROMPT analysis}
\label{table:var_names}
\end{table}

\end{document}